\documentclass[a4paper,UKenglish]{lipics-v2019}
\usepackage{amsmath}
\usepackage{algorithm}
\usepackage[noend]{algpseudocode}
\usepackage{amsthm}
\usepackage{float}
\usepackage{graphicx}
\usepackage[utf8]{inputenc}
\usepackage{subcaption}
\usepackage{nicefrac}
\usepackage{wrapfig}
\usepackage[dvipsnames]{xcolor}
\input{insbox}

\newcommand{\idlow}[1]{\mathord{\mathcode`\-="702D\it #1\mathcode`\-="2200}}
\newcommand{\idl}[1]{\ensuremath{\idlow{#1}}}

\def\listnameplural{splay-lists}
\def\listname{splay-list}
\def\Listname{Splay-List}
\def\listcodename{splaylist}
\def\Listcodename{Splaylist}

\title{The \Listname: A Distribution-Adaptive Concurrent Skip-List}
\author{Vitaly Aksenov}{ITMO University}{aksenov.vitaly@gmail.com}{}{}
\author{Dan Alistarh}{IST Austria}{dan.alistarh@ist.ac.at}{}{}
\author{Alexandra Drozdova}{ITMO University}{drsanusha1@gmail.com}{}{}
\author{Amirkeivan Mohtashami}{Sharif University}{akmohtashami97@gmail.com}{}{}
\authorrunning{V. Aksenov, D. Alistarh, A. Drozdova and A. Mohtashami}



\newcommand{\myparagraph}[1]{\noindent\textbf{#1}}

\newboolean{showcomments}
\setboolean{showcomments}{true}
\ifthenelse{\boolean{showcomments}}
{ \newcommand{\mynote}[3]{
    \fbox{\bfseries\sffamily\scriptsize#1}
    {\small$\blacktriangleright$\textsf{\emph{\color{#3}{#2}}}$\blacktriangleleft$}}
}
{
\newcommand{\mynote}[3]{}}

\newcommand{\EE}{\mathbb{E}}

\newcommand{\dan}[1]{\mynote{Dan}{#1}{green}}
\newcommand{\akm}[1]{\mynote{Keivan}{#1}{blue}}

\usepackage{listings,multicol}
\begin{document}
\nolinenumbers

\maketitle

\begin{abstract}
The design and implementation of efficient concurrent data structures has seen significant attention. 
However, most of this work has focused on concurrent data structures providing good \emph{worst-case} guarantees. 
In real workloads, objects are often accessed at different rates, since access distributions may be non-uniform. Efficient distribution-adaptive data structures are known in the sequential case, e.g. the splay-trees; however, they often are hard to translate efficiently in the concurrent case.    

In this paper, we investigate distribution-adaptive concurrent data structures, and propose a new design called the splay-list. At a high level, the splay-list is similar to a standard skip-list, with the key distinction that the height of each element adapts dynamically to its access rate: popular elements ``move up,'' whereas rarely-accessed elements decrease in height. We show that the splay-list provides order-optimal amortized complexity bounds for a subset of operations, while being amenable to efficient concurrent implementation. 
Experimental results show that the splay-list can leverage distribution-adaptivity to improve on the performance of classic concurrent designs, and can outperform the only  previously-known distribution-adaptive design in certain settings. 

\end{abstract}


\section{Introduction}

The past decades have seen significant effort on designing efficient concurrent data structures, leading to fast variants being known for many classic data structures, such as hash tables, e.g.~\cite{Michael, HSBook}, skip lists, e.g.~\cite{FraserThesis, LazySkiplist, DougSkiplist}, or search trees, e.g.~\cite{Ellen10, Natarajan14}. 
Most of this work has focused on efficient concurrent variants of data structures with optimal \emph{worst-case} guarantees. 
However, in many real workloads, the access rates for individual objects are not uniform. This fact is well-known, and is modelled in several industrial benchmarks, such as YCSB~\cite{cooper2010benchmarking}, or TPC-C~\cite{TPCC}, where the generated access distributions are heavy-tailed, e.g., following a Zipf distribution~\cite{cooper2010benchmarking}. 
While in the sequential case the question of designing data structures which adapt to the access distribution is well-studied, see e.g.~\cite{knuth1997art} and references therein, 
in the concurrent case significantly less is known. 
The intuitive reason for this difficulty is that self-adjusting data structures require non-trivial and frequent pointer manipulations, such as node rotations in a balanced search tree, which can be complex to implement concurrently. 

To date, the CBTree~\cite{CBTree} is the only concurrent data structure which leverages the  skew in the access distribution for faster access. At a high level, the CBTree is a concurrent search tree maintaining internal balance with respect to the access statistics per node. 
Its sequential variant provides order-optimal amortized complexity bounds (static optimality), and empirical results show that it  provides significant performance benefits over a classic non-adaptive concurrent design for skewed workloads.  
At the same time, the CBTree may be seen as fairly complex, due to the difficulty of re-balancing in a concurrent setting, and the paper's experimental validation  suggests that maintaining exact access statistics and balance in a concurrent setting come at some performance cost---thus, the authors propose a limited-concurrency variant, where  rebalancing is delegated to a single thread. 

In this paper, we revisit the topic of distribution-adaptive concurrent data structures, and propose a design called the \emph{splay-list}. 
    At a very high level, the splay-list is very similar to a classic skip-list~\cite{Pugh}: 
    it consists of a sequence of sorted lists, ordered by containment, where the bottom-most list contains all the elements present, and each higher list contains a sub-sample of the elements from the previous list. 
    The crucial distinction is that, in contrast to the original skip-list, where the height of each element is chosen randomly, in the splay-list, the height of each element \emph{adapts} to its access rate: elements that are accessed more often move ``up,'' and will be faster to access, whereas elements which are accessed less often are demoted towards the bottom-most list. 
    Intuitively, this property ensures that popular elements are closer to the ``top'' of the list, and are thus accessed more efficiently. 
    
    This intuition can be made precise: we provide a rebalancing algorithm which ensures that, after $m$ operations, the amortized search and delete time for an item $x$ in a sequential splay-list is $\mathcal{O}\left(\log\frac{m}{f(x)}\right)$ where $f(x)$ is the number of previous searches for $x$, whereas insertion takes amortized $
    \mathcal{O}(\log m)$ time. This asymptotically matches the guarantees of the CBTree~\cite{CBTree}, and implies static optimality.  
    Since maintaining exact access statistics for each object can hurt performance---as every search has to write---we introduce and present guarantees for variants of the data structure which only maintains \emph{approximate} access counts. 
    If rebalancing is only performed with probability $1 / c$---meaning that only this fraction of readers will have to write---then we show that the expected amortized cost of a contains operation becomes $\mathcal{O}\left(c \log \frac{m}{f(x)} \right)$. Since $c$ is a constant, this trade-off can be beneficial. 
    
    From the perspective of concurrent access, an advantage of the splay-list is that it can be easily implemented on top of existing skip-list designs~\cite{HSBook}: 
    the pointer changes for promotion and demotion of nodes are operationally a subset of skip-list insertion and deletion operations~\cite{Fraser}. 
    At the same time, our design does come with some limitations: (1) since it is based on a skip-list backbone, the splay-list may have higher memory cost and path length relative to a tree; (2) as discussed above, approximate access counts are necessary for good performance, but come at an increase in amortized expected cost, which we believe to be inherent; (3) for simplicity, our update operations are lock-based (although this limitation could be removed). 
    
    We implement the splay-list in C++ and compare it with the CBTree and a regular skip-list on uniform and skewed workloads, and for different update rates. 
    Overall results show that the \listname{} can indeed leverage workload skew for higher performance, and that it can scale when access counts are approximate.   
    By comparison, the CBTree also scales well for moderately skewed workloads and low update rates, in which case it outperforms the \listname{}. However, it has relatively lower performance for moderate or high update rates. We recall that the original CBTree paper proposes a practical implementation with limited concurrency, in which all rebalancing is performed by a single thread. 

Overall, the results suggest a trade-off between the performance of the two data structures and the workload characteristics, both in terms of access distribution and access types. 
The fact that the \listname{} can outperform the CBTree in some practical scenarios may appear surprising, given that the \listname{}  leads to longer access paths on average due to its skip-list backbone. 
However, our design benefits from allowing additional concurrency, and the caching mechanism serves to hide some of the additional access costs.

    \myparagraph{Related Work.} 
    The literature on \emph{sequential} self-adjusting  data structures is well-established, and extremely vast. 
    We therefore do not attempt to cover it in detail, and instead point the reader to classic texts, e.g.~\cite{knuth1997art, sleator1985self} for details. 
    Focusing on self-adjusting skip-lists, we note that statically-optimal \emph{deterministic} skip-list-like data structures can be derived from the $k$-forest structure of Martel~\cite{martel1991self}, or from the working set structure of Iacono~\cite{iacono2001alternatives}. 
    Ciriani et al.~\cite{ciriani2002static} provide a similar randomized approach for constructing a self-adjusting skip-list for string dictionary operations in the external memory model. 
    Bagchi et al.~\cite{bagchi2005biased} introduced a general \emph{biased skip-list} data structure, which maintains balance w.r.t. node height when nodes can have arbitrary weight, while Bose et al.~\cite{bose2008dynamic} built on biased skip-lists  to obtain a \emph{dynamically-optimal} skip-list data structure. 
    
    Relative to our work, we note that, naturally, the above theoretical references provide stronger guarantees relative to the \listname{} in the sequential setting.
    At the same time, they are quite complex, and would not extend efficiently to a concurrent setting. 
    Two practical additions that our design brings relative to this prior work is that we are the first to provide bounds even when the access count values are \emph{approximate} (Section~\ref{sec:relaxation-theory}), and that our concurrent design allows the \listname{} adjustment to occur in a single pass (Section~\ref{sec:practical}). 
    Reference~\cite{CBTree} posed the existence of an efficient self-balancing skip-list variant as an open question---we answer this question here, in the affirmative. 
    
    The \listname{} ensures similar complexity guarantees as the CBTree~\cite{CBTree}, although its structure is different. 
    Both references provide complexity guarantees under \emph{sequential} access. 
    In addition, we provide complexity guarantees in the case where the access counts are maintained via \emph{approximate} counters, in which case the CBTree is not known to provide guarantees.  
    One obvious difference relative to our work is that we are investigating a skip-list-based design.
    This allows for more concurrency: the proposed practical implementation in~\cite{CBTree} assumes that adjustments are performed only by a dedicated thread, whereas \listname{} updates can be performed by any thread. At the same time, our design shares some of the limitations of skip-list-based data structures, as discussed above. 
    
    There has been a significant amount of work on efficient concurrent ordered maps, see e.g.~\cite{brown2017techniques, 215945} for an overview of recent work. 
    However, to our knowledge, the CBTree remained the only non-trivial self-adjusting concurrent data structure.

\section{The Sequential \Listname{}}

The \listname{} design builds on the classic skip-list by Pugh~\cite{Pugh}. In the following, we will only briefly overview the skip-list structure, and focus on the main technical differences. We refer the reader to~\cite{HSBook} for a more in-depth treatment of concurrent  skip-lists. 

\myparagraph{Preliminaries.} 
Similar to skip-lists, the \listname{} maintains a set of sorted lists, starting from the bottom list,  which contains all the objects present in the data structure. Without loss of generality, we assume that each object consists of a key-value pair. 
We thus use the terms \emph{object} and \emph{key} interchangeably.  
It is useful to view these lists as stacked on top of each other; a list's index (starting from the bottom one, indexed at $0$) is also called its \emph{height}. 
The lists are also ordered by containment, as a higher-index list contains a subset of the objects present in a lower-index list. The higher-index lists are also called \emph{sub-lists}.
The bottom list, indexed at $0$, contains all the objects present in the data structure at a given point in time. Unlike skip-lists, where the choice of which objects should be present in each sub-list is random, a \listname's structure is adjusted according to the access distribution across keys/objects. 

The following definitions make it easier to understand how the operations are handled in \listnameplural. The \emph{height of the \listname} is the number of its sub-lists. The \emph{height of an object} is the height of the highest sub-list containing it. Typically, we do not distinguish between the object and its key. 
The height of a key $u$ is the height of a corresponding object $h_u$. 
Key $u$ is the \emph{parent of key $v$ at height $h$} if $u$ is the largest key whose value is smaller than or equal to $v$, and whose  height is at least $h$. That is, $u$ is the last key at height $h$ in the traversal path to reach $v$. 
Critically, note that, if the height of a key $v$ is at least $h$, then $v$ is its own parent at height $h$; otherwise, its parent is some node $v \neq u$.
In addition, we call the set of objects for which $u$ is the parent at height $h$, its \emph{$h$-children} or the \emph{subtree of $u$ at height $h$, denoted by $C^h_u$}.

Our data structure supports three standard methods: \texttt{contains}, \texttt{insert} and \texttt{delete}.
We say that a \texttt{contains} operation is \emph{successful} (returns \emph{true})
if the requested key is found in the data structure and was not marked as deleted; otherwise, the operation is \emph{unsuccessful}.
An \texttt{Insert} operation is \emph{successful} (returns \emph{true}) if the requested key was not present upon insertion; otherwise, it is \emph{unsuccessful}.
A \texttt{Delete} operation is \emph{successful} (returns \emph{true}) if the requested key is found and was not marked as deleted, otherwise, the operation is \emph{unsuccessful}.
As suggested, in our implementation the \texttt{delete} implementation does not always unlink the object from the lists--instead, it may just mark it as deleted.

For every key $u$, we maintain a counter $hits_u$, which counts the number of \texttt{contains($u$)}, \texttt{insert($u$)}, and \texttt{delete($u$)} operations which \emph{visit the object}.
In particular, \emph{successful} \texttt{contains($u$)}, \texttt{insert($u$)}, and \texttt{delete($u$)} operations increment $hits_u$ 
Moreover, unsuccessful operations can also increment $hits_u$ if the element is physically present in the data structure, even though logically deleted, upon the operation. 
In this case, the marked element is still visited by the corresponding operation.  
(We will re-discuss this notion in the later sections, but the simple intuition here is that we cannot store access counts for elements which are not physically present in the data structure, and therefore ignore their access counts.) 
We will refer to operations that visits an object with the corresponding key simply as \emph{hit-operations}.


For any set of keys $S$, we define a function $hits(S)$ to be the sum of the number of hits-operations performed to the keys in $S$.
As usual, sentinel \emph{head} and \emph{tail} nodes are added to all sub-lists. 
The height of a sentinel node height is equal to the height of the \listname{} itself, and exceeds the height of all other nodes by at least $1$. By convention, $hits_{head} = hits_{tail} = 1$.

\subsection{The \texttt{contains} Operation} 

\myparagraph{Overview.} The contains operation consists of two phases: the search phase and the balancing phase.
The search phase is exactly as in skip-list: starting from the head of the top-most list, we traverse the current list until we find the last object with key lower than or equal to the search key. If this object's key is not equal to the search key, the search continues from the same object in the lower list. Otherwise, the search operation completes. The process is repeated until either the key is found or the algorithm attempts to descend from the bottom list, in which case the key is not present.

If the operation finds its target object, its \emph{hits} counter is incremented and the balancing phase starts: its goal is to update the \listname's structure to better fit the access distribution, by traversing the search path backwards and checking two conditions, which we call the \emph{ascent} and \emph{descent} conditions.

We now overview these conditions. For the descent condition, consider two neighbouring nodes at height $h$, corresponding to two keys $v < u$. Assume that both $v$ and $u$ are on level $h$, and consider their respective subtrees $C^h_v$ and $C^h_u$. 
Assume further that the number of hits to objects in their subtrees ($hits(C^h_v \cup C^h_u)$) became smaller than a given threshold, which we deem appropriate for the nodes to be at height $h$. (This threshold is updated as more and more operations are performed.)
To fix this imbalance, we can ``merge'' these two subtrees, by descending the right neighbour, $u$, below $v$, thus creating a new subtree of higher overall hit count. 
Similarly, for the ascent condition, we check whether an object's subtree has \emph{higher} hit count than a threshold, in which case we increase its height by one. 

Now, we describe the conditions more formally.
Assume that the total number of hit-operations to all objects, including those marked for deletion, appearing in \listname{} is $m$, and that the current height of the \listname{} is equal to $k+1$. Thus, there are $k$ sub-lists, and the sentinel sub-list containing exclusively $head$ and $tail$. Excluding the head, for each object $u$ on a backward path, the following conditions are checked in order.

\myparagraph{The Descent Condition.} Since $u$ is not the head, there must exist an object $v$ which precedes it in the forward traversal order, such that $v$ has height  $\geq h_u$. 
If $$hits(C^{h_u}_u) + hits(C^{h_u}_v) \leq \frac{m}{2^{k - h_u}},$$ then the object $u$ is demoted from height $h_u$, by simply being removed from the sub-list at height $h_u$. The object stays a member of the sub-list at height $h_u - 1$ and $h_u$ is decremented.  The backward traversal is then continued at $v$. 

\myparagraph{The Ascent Condition.} Let $w$ be the first successor of $u$ in the list at height $h_u$, such that $w$ has height \emph{strictly greater than} $h_u$. Denote the set of objects with keys in the interval $[u, w)$ with height equal to $h_u$ by $S_u$. If the number of hits $m$ is greater than zero and the following inequality holds: $$ \sum\limits_{x \in S_u} hits(C^{h_u}_x) > \frac{m}{2^{k - h_u - 1}},$$ then $u$ is promoted and inserted into the sub-list at height $h_u + 1$.
The backward traversal is then continued from $u$, which is now in the higher-index sub-list. 
The rest of the path at height $h_u$ is skipped. 
Note that the object $u$ is again checked against the ascent condition at height $h_u + 1$, so it may be promoted again. 
Also note that the calculated sum is just an interval sum, which can be maintained efficiently, as we show later.

\myparagraph{\Listname{} Initialization and Expansion.} 
Initially, the \listname{} is empty and has only one level with two nodes, head and tail. Suppose that the total number of hits to objects in \listname{} is $m$. The lowest level on which the object can be depends on how low the element can be demoted. Suppose that the current height of the list is $k + 1$. Consider any object at the lowest level $0$: in the descent condition we compare $hits(C^0_u) + hits(C^0_v)$ against $\frac{m}{2^{k}}$. While $m$ is less than $2^{k + 1}$, the object cannot satisfy this condition since $C^{h_u}_v \geq hits_v \geq 1$, but when $m$ becomes larger than this threshold, it could. Thus, we have to increase the height of \listname{} and add a new list to allow such an object to be demoted. 
By that, the height of the \listname{} is always $\log m$.
This process is referred to as \emph{\listname{} expansion}. 
Notice that this procedure could eventually lead to a skip-list of unbounded height.
However, this height does not exceed $64$, since this would mean that we performed at least $2^{64}$ successful operations which is unrealistic. 
We discuss ways to make this procedure more practical, i.e., lazily increase the height of an object only on its traversal, in Section~\ref{sec:practical}.

\myparagraph{The Backward Pass.} Now, we return to the description of the \texttt{contains} function. The first phase is the forward pass, which is simply the standard search algorithm which stores the traversal path. If the key is not found, then we stop. Otherwise, suppose that we found an object $t$. We have to restructure the \listname{} by applying ascent and descent conditions. Note, that the only objects that are affected and can change their height lie on the stored path. For that, in each object $u$ we store the total hits to the object itself, $hits_u$, as well as the total number of hits into the ``subtree'' of each height excluding $u$, i.e., for all $h$ we maintain $hits^h_u = hits(C^h_u \setminus \{u\})$. We denote the hits to the object $u$ as $sh_u$.

\noindent Thus, when traversing the path backwards and we check the following:
\begin{enumerate}
\item If the object $u \ne t$ is a parent of $t$ on some level $h$, we then increase its $hits^h_u$ counter. Note that $h \leq h_u$.
\item Check the descent condition for $v$ and $u$ as $sh_v + hits^{h_u}_v + sh_u + hits^{h_u}_u \leq \frac{m}{2^{k - h_u}}$. If this is satisfied, demote $u$ and increment $hits^{h_u}_v$ by $sh_u + hits^{h_u}_u$. Continue  on the path.
\item Check the ascent condition for $u$ by comparing $\sum_{w \in S_u} sh_w + hits^{h_u}_w$ with $\frac{m}{2^{k - h_u - 1}}$. If this is satisfied, add $u$ to the sub-list $h_u + 1$, set $hits^{h_u + 1}_u$ to the calculated sum minus $sh_u$ and decrease $hits^{h_u + 1}_v$ by the calculated sum, where $h$ is a parent of $u$ at height $h_u + 1$. We then continue with the sub-list on level $h_u + 1$. Below, we describe how to maintain this sum in constant time.
\end{enumerate}

\myparagraph{The partial sums trick.} Suppose that $p(u)$ is the parent of $u$ on level $h_u + 1$. During the forward pass, we compute the sum of $hits(C^{h_u}_x) = sh_x + hits^{h_u}_x$ over all objects $x$ which lie on the traversal path between $p(u)$ (including it) and $u$ (not including it). Denote this sum by  $P_u$. Thus, to check the ascent condition on the backward pass, we simply have to compare $\sum\limits_{x \in S_u} sh_u + hits(C^{h_u}_x) = sh_{p(u)} + hits^{h_u + 1}_{p(u)} - P_u$ against $\frac{m}{2^{k - h_u - 1}}$. Observe that the partial sums $hits(S_u)$ can be increased only by one after each operation. Thus, the only object on level $h$ that can be promoted is the leftmost object on this level. For the first object $u$, $S_u$ can be calculated as $hits^{h_u + 1}_{p(u)} - hits^{h_u}_{p(u)}$. In addition, after the promotion of $u$, only $u$ and $p(u)$ have their $hits^{h_u + 1}$ counters changed. Moreover, there is no need to skip the objects to the left of the promoted object, as suggested by the ascent condition, since there cannot be any such objects.

\myparagraph{Example.}
To illustrate, consider the \listname{} provided on Figure~\ref{fig:example:start}. It contains keys $1, \ldots, 6$ with values $m = 10$ and $k = \lfloor \log m \rfloor = 3$. We can instantiate the sets described above as follows: $C^1_3 = \{3, 4, 5\}$, $C^1_2 = \{2\}$, $C^1_{head} = \{head, 1\}$ and $C^2_{head} = \{head, 1, 2, \ldots, 5\}$. At the same time, $S_4 = \{4, 5\}$, $S_3 = \{3\}$ and $S_2 = \{2, 3\}$.
In the Figure, the cell of $u$ at height $h > 0$ contains $hits^h_u$, while the cell at height $0$ contains $sh_u$. For example, $sh_3 = 1$ and $hits^1_3 = sh_4 + sh_5 = 2$, $sh_2 = 1$ and $hits^1_2 = 0$, $sh_1 = 1$ and $hits^2_{head} = 5$.

Assume we execute \texttt{contains($5$)}. 
On the forward path, we find $5$ and the path to it is $2 \rightarrow 3 \rightarrow 4 \rightarrow 5$. We increment $m$, $sh_5$, $hits^1_3$ and $hits^2_{head}$ by one. Now, we have to adjust our \listname{} on the backward path. We start with $5$: we check the descent condition by comparing $hits(C^0_4) + hits(C^0_5) = 3$ with $\frac{m}{2^{k - 0}} = \frac{11}{8}$ and the ascent condition by comparing $hits(S_5) = 2$ with $\frac{m}{2^{k - 0 - 1}} = \frac{11}{4}$. Obviously, neither condition is satisfied. We continue with $4$: the descent condition by comparing $hits(C^0_3) + hits(C^0_4) = 2$ with $\frac{11}{8}$ and the ascent condition by comparing $hits(S_4) = 3$ with $\frac{11}{4}$~--- the ascent condition is satisfied and we promote object $4$ to height $1$ and change the counter $hits^1_3$ to $2$. For $3$, we compared $hits(C^1_2) + hits(C^1_3) = 2$ with $\frac{11}{4}$ and $hits(S_3) = 4$ with $\frac{11}{2}$~--- the descent condition is satisfied and we demote object $3$ to height $0$ and change the counter $hits^1_2$ to $1$. Finally, for $2$ we compared $hits(C^1_1) + hits(C^1_2) = 4$ with $\frac{11}{4}$ and $hits(S_2) = 5$ with $\frac{11}{2}$~--- none of the conditions are satisfied. As a result we get the \listname{} shown on Figure~\ref{fig:example:result}.

\begin{figure*}
\centering
\begin{subfigure}{.5\linewidth}
\includegraphics[width=\linewidth]{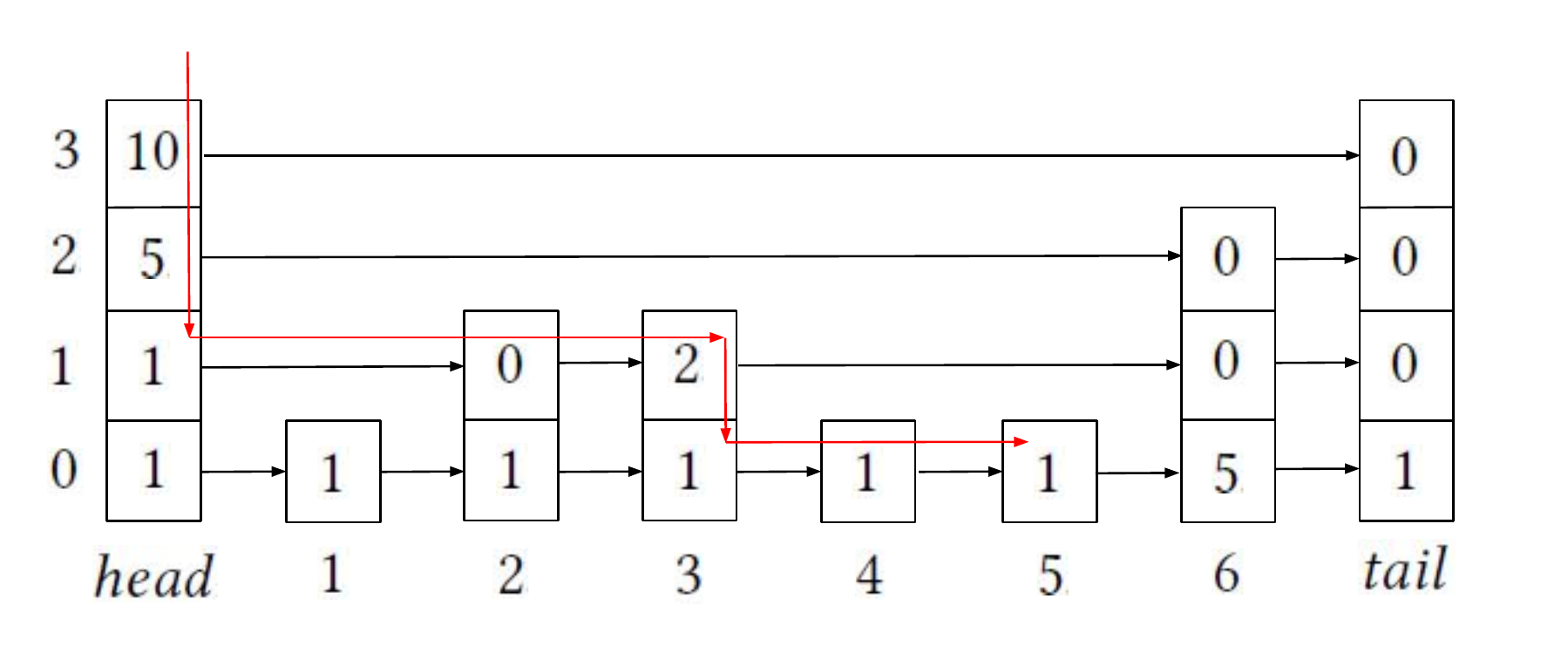}
\caption{Before \texttt{contains($5$)}}
\label{fig:example:start}
\end{subfigure}%
\begin{subfigure}{.5\linewidth}
\includegraphics[width=\linewidth]{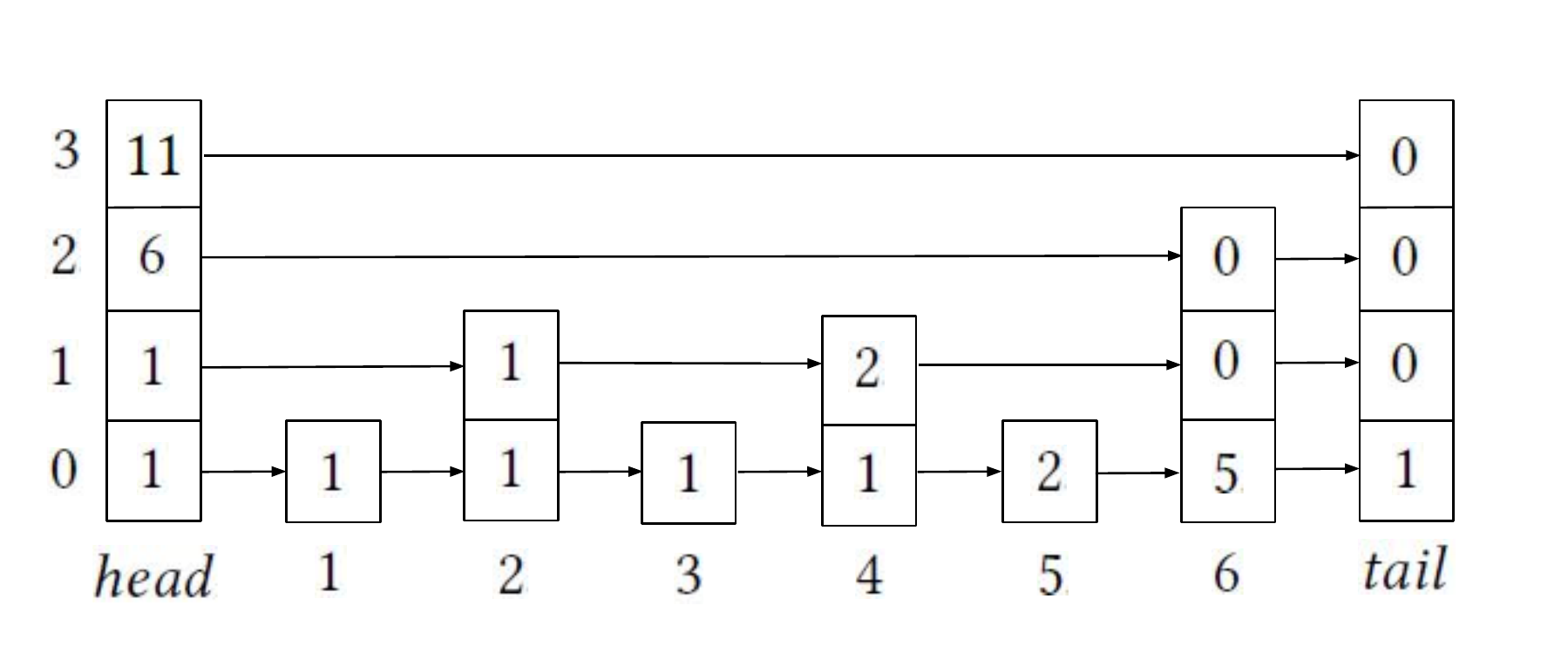}
\caption{After \texttt{contains($5$)}}
\label{fig:example:result}
\end{subfigure}
\caption{Example of \listname{}}
\end{figure*}

\subsection{\texttt{Insert} and \texttt{Delete} operations}

\myparagraph{Insertion.} Inserting a key $u$ is done by first finding the object with the largest key  lower than or equal to $u$. In case an object with the key is found, but is marked as logically deleted, the insertion unmarks the object, increases its hits counter and completes successfully. Otherwise, $u$ is inserted on the lowest level after the found object. 
This item has hits count set to $1$. In both cases, the structure has to be re-balanced on the backward pass as in \texttt{contains} operation. Unlike the skip-list, \listnameplural{} always physically inserts into the lowest-level list.

\myparagraph{Deletion.}
This operation needs additional care. The operation first searches for an object with the specified key. If the object is found, then the operation logically deletes it by marking it as \texttt{deleted}, increases the hits counter and performs the backward pass. Otherwise, the operation completes.

Notice that we maintain the total number of hits on currently logically deleted objects. When it becomes at least half of $m$, the total number of hits to all objects, we initialize a new structure, and move all non-deleted objects with corresponding hits to it.

\myparagraph{Efficient Rebuild.} The only question left is how to build a new structure efficiently enough to amortize the performed delete operations.
Suppose that we are given a sorted list of $n$ keys $k_1, \ldots, k_n$ with the number of hit-operations on them $h_1, \ldots, h_n$, where their sum is equal to $M$. We propose an algorithm that builds a \listname{} such that no node satisfies the ascent and descent conditions, using $O(M)$ time and $O(n \log M)$ memory.

The idea behind the algorithm is the following. 
We provide a recursive procedure that takes the contiguous segment of keys $k_l, \ldots, k_r$ with the total number of accesses $H = h_l + \ldots + h_r$. The procedure finds $p$ such that $2^{p - 1} \leq H < 2^p$. Then, it finds a key $k_s$ such that $h_l + \ldots + h_{s - 1}$ is less than or equal to $\frac{H}{2}$ and $h_{s + 1} + \ldots + h_r$ is less than $\frac{H}{2}$. We create a node for the key $k_s$ with the height $p$, and recursively call the procedure on segments $k_l, \ldots, k_{s - 1}$ and $k_{s + 1}, \ldots, k_r$. There exists a straightforward implementation which finds the split point $s$ in $O(r - l)$, i.e., linear time. The resulting algorithm works in $O(n \log M)$ time and takes $O(n \log M)$ memory: the depth of the recursion is $\log M$ and on each level we spend $O(n)$ steps.

However, the described algorithm is not efficient if $M$ is less than $n \log M$. 
To achieve $O(M)$ complexity, we would like to answer the query to find the split point $s$ in $O(1)$ time. 
For that, we prepare a special array $T$ which contains in sorted order $h_1$ times key $k_1$, $h_2$ times key $k_2$, $\ldots$, $h_n$ times key $k_n$. 
To get the required $s$, at first, we take a subarray of $T$ that corresponds to the segment $[l, r]$ under the process, i.e., $h_l$ times key $k_l$, $\ldots$, $h_r$ times key $k_r$. Then, we take the key $k_i$ that is located in the middle cell $\lceil \frac{h_l + \ldots + h_r}{2} \rceil$ of the chosen subarray. This $i$ is our required $s$.
Let us calculate the total time spent: the depth of the recursion is $\log M$; there is one element on the topmost level which we insert in $\log M$ lists, there are at most two elements on the next to topmost level which we insert in $\log M - 1$ lists, and etc., there are at most $2^i$ elements on the $i$-th level from the top which we insert in $\log M - i$ lists. The total sum is clearly $O(M)$.

Thus, the final algorithm is: if $M$ is larger than $n \log M$, then we execute the first algorithm, otherwise, we execute the second algorithm. The overall construction works in $O(M)$ time and uses $O(n \log M)$ memory.

\section{ Sequential \Listname{} Analysis}

\myparagraph{Properties.} We begin by stating some invariants and general propertties of the \listname{}. 


\begin{lemma}
\label{lemma:no-ascent}
After each operation, no object can satisfy the ascent condition.
\end{lemma}
\begin{proof}
Note that we only consider the hit-operations, i.e., the operations that change $hits$ counters, because other operations do not affect any conditions.
We will proceed by induction on the total number $m$ of hit-operations on the objects of \listname{}.

For the base case $m = 0$, the \listname{} is empty and the hypothesis trivially holds.
For the induction step, we assume that the hypothesis holds before the start of the $m$-th operation, and we verify that it holds after the operation completes.

First, recall that, for a fixed object $u$, the set $S_u$ is defined to include all objects of the same height between $u$ and the successor of $u$ with height \emph{greater} than $h_u$.
Specifically, we name the sum $\sum\limits_{x \in S_u} hits(C^h_x)$ in the ascent condition as the object $u$'s \textbf{ascent potential}. Note that after the forward pass and the increment of $sh_u$ and $hits^h_v$ counters where $v$ is a parent of $u$ on height $h$, only the objects on the path have their ascent potential increased by one and, thus, only they can satisfy the ascent condition.

Now, consider the restructuring done on the backward pass. If the object $u$ satisfies the descent condition, i.e., $v$ precedes $u$ and $T = hits(C^{h_u}_v) + hits(C^{h_u}_u) \leq \frac{m}{2^{k - h}}$, we have to demote it. After the descent, the ascent potential of the objects between $v$ and $u$ on the lower level $h_u - 1$ have changed. However, these potentials cannot exceed $T$, meaning that these objects cannot satisfy the ascent condition.

Consider the backward pass, and focus on the set of objects at height $h$. We claim that only the leftmost object at that height can be promoted, i.e., its preceding object has a height greater than $h$. This statement is proven by induction on the backward path. Suppose that we have $\ell$ objects with height $h$ on the path, which we denote by $u_1, u_2, \ldots, u_{\ell}$. By induction, we know that none of the objects on the path with lower height can ascend higher than $h$: these objects appear to the right of $u_1$. We know that each object was accessed at least once, $sh_{u_i} \geq 1$, and, thus, we can guarantee that $hits(S_{u_1}) > hits(S_{u_2}) > \ldots > hits(S_{u_{\ell}})$. Since the ascent potentials $hits(S_{u_i})$ are increased only by one per operation, the first and the only object that can satisfy the ascent condition is $u_1$, i.e., the leftmost object with the height $h$. If it satisfies the condition, we promote it. Consider the predecessor of $u_1$ on the forward path: the object $v$ with height $h_v > h$. Object $u_1$ can be promoted to height $h_v$, but not higher, since the ascent potential of the objects on the path with height $h_v$ does not change after the promotion of $u$, and only the leftmost object on that level can ascend. However, note that $hits^{h_v}_v$ can decrease and, thus, it can satisfy the descent condition, while $u_1$ cannot since $hits^h_{u_1}$ was equal to $hits(S_{u_1})$ before the promotion and it satisfied the ascent condition.

Because the only objects that can satisfy the ascent condition lie on the path, and we promoted necessary objects during the backward pass, no object may satisfy the ascent condition at the end of the traversal. That is exactly what we set out to prove.\end{proof}

\begin{lemma}
\label{lemma:max-sublists}
Given a hit-operation with argument $u$, the number of sub-lists visited during the forward pass is at most $3 + \log\frac{m}{sh_u}$.
\end{lemma}
\begin{proof}
During the forward pass the number of hits does not change; thus, according to Lemma \ref{lemma:no-ascent}, the ascent condition does not hold for $u$. Hence $sh_u \leq \frac{m}{2^{k - h_u - 1}}$. We get that $k - h_u - 1 \leq \log\frac{m}{sh_u}$. Since during the forward pass $(k + 1) - h_u + 1$ sub-lists are visited (notice the sentinel sub-list), the claim follows.
\end{proof}

\begin{lemma}
\label{lemma:forward-4}
In each sub-list, the forward pass visits at most four objects that do not satisfy the descent condition.
\end{lemma}

\begin{proof}
Suppose the contrary and that the algorithm visits at least five objects $u_1, u_2, \ldots, u_5$ in order from left to right, that do not satisfy the descent condition in sub-list $h$. The height of the objects $u_2, \ldots, u_5$ is $h$, while the height of $u_1$ might be higher. See Figure~\ref{fig:forward-lemma}.

\InsertBoxR{0}{\begin{minipage}{0.5\linewidth}
    \centering
    \includegraphics[width=\linewidth]{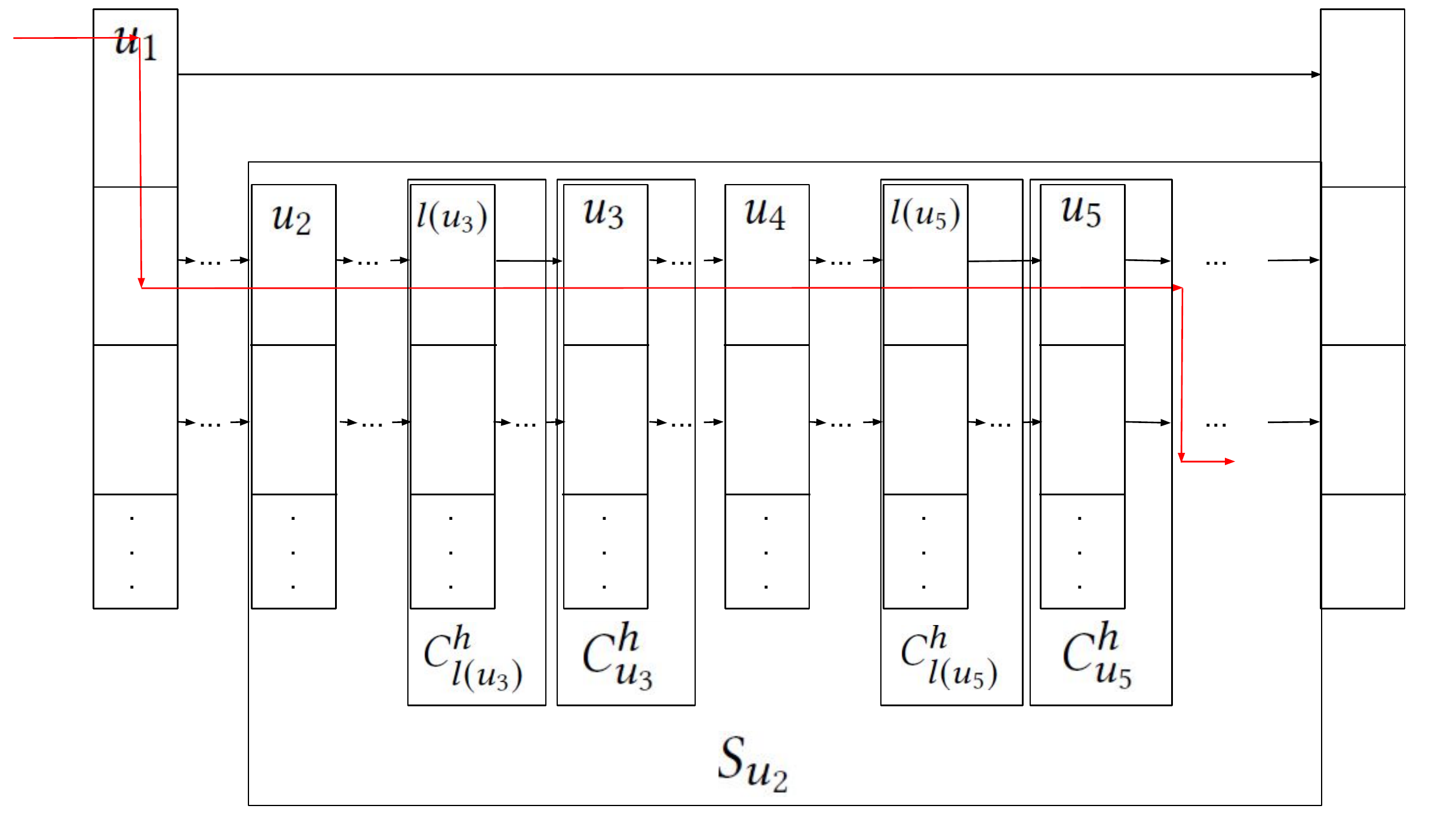}
    \captionof{figure}{Depiction of the proof of Lemma 3}
    \label{fig:forward-lemma}
\end{minipage}%
}[4]

Note that if the descent condition does not hold for an object $u$, the demotion of another object of the same height cannot make the descent condition for $u$ satisfiable. Therefore, since the condition is not met for $u_3$ and $u_5$, the sum $hits(S_{u_2}) \geq (hits(C^h_{l(u_3)}) + hits(C^h_{u_3})) + (hits(C^h_{l(u_5)}) + hits(C^h_{u_5})) > \frac{m}{2^{k - h}} + \frac{m}{2^{k - h}} = \frac{m}{2^{k - h - 1}}$, where $l(u_3)$ and $l(u_5)$ are the predecessors of $u_3$ and $u_5$ on height $h$. Note that it is possible that $l(u_3)$ and $l(u_5)$ would be the same as $u_2$ and $u_4$ respectively. This means that $u_2$ satisfies the ascent condition, which contradicts Lemma \ref{lemma:no-ascent}. 

Note that we considered four objects since $u_1$ is an object of height greater than $h$. 
\end{proof}

Since only the leftmost object can be promoted, the backward path coincides with the forward path. Thus, the following lemma trivially holds.
\begin{lemma}
\label{lemma:backward-4}
During the backward pass, in each sub-list $h$, at most four objects are visited that do not satisfy the descent condition.
\end{lemma}

\begin{theorem}
\label{theorem:total-length}
If $d$ descents occur when accessing object $u$, the sum of the lengths of the forward and backward paths is at most $2d + 8y$, where $y = 3 + \log\frac{m}{sh_u}$.
\end{theorem}
\begin{proof}
Each object satisfying the descent condition is passed over twice, once in the forward and again in the backward pass. According to Lemma \ref{lemma:max-sublists}, there are at most $y$ sub-lists that are visited during either passes. Excluding the descended objects, the total length of the forward path, according to Lemma \ref{lemma:forward-4} is $4y$. Lemma \ref{lemma:backward-4} gives the same result for the backward path. Hence, the total length is $2d + 8y$ which is the desired result.
\end{proof}

\myparagraph{Asymptotic analysis.} 
We can now finally state our main analytic result. 

\begin{theorem}
\label{thm:main-asymptotics}
The hit-operations with argument $u$ take amortized  $O\left(\log\frac{M}{sh_u}\right)$ time, where $M$ is the total number of hits to non-marked objects of the \listname.
At the same time, all other operations take amortized $O(\log M)$ time.
\end{theorem}
\begin{proof}
We will prove the same bounds but with $m$ instead of $M$. Please note that since we rebuild the \listname{} is triggered when $M$ becomes less than $\frac{m}{2}$, we can always assume that $M \geq \frac{m}{2}$ and, thus, the bounds with $m$ and $M$ differ only by a constant. 

First, we deal with the \listname{} expansion procedure: it adds only $O(1)$ amortized time to an operation. The expansion happens when $m$ is equal to the power of two and costs $O(m)$. Since, from the last expansion we performed at least $\frac{m}{2}$ hits operations 
we can amortize the cost $O(m)$ against them. Note that each operation will be amortized against only once, thus the amortization increases the complexity of an operation only by $O(1)$.

Since the primitive operations such as following the list pointer, a promotion with the ascent check and a demotion with the descent check are all $O(1)$, the cost of an operation is in the order of the length of the traversed path. 
According to Theorem~\ref{theorem:total-length}, the total length of the traversed path during an operation is $2 \cdot d + 8 \cdot y$ where $d$ is the number of vertices to demote and $y$ is the number of traversed layers: if the object $u$ was found $y$ is equal to $O\left(\log\frac{m}{sh_u}\right)$, otherwise, it is equal to $\log m$, the height of the \listname.

Note that the number of promotions per operation cannot exceed the number of passed levels $y$, since only one object can satisfy the ascent condition per level. At the same time, the total number of demotions across all operations, i.e., the sum of all $d$ terms, cannot exceed the total number of promotions. Thus, the amortized time of the operation can be bounded by $O(\text{number of levels passed})$ which is equal to what we required.

The amortized bound for \texttt{delete} operation needs some additional care. The operation can be split into two parts: 1)~find the object in the \listname, mark it as deleted and adjust the path; 2)~the reconstruction part when the object is physically deleted. The first part is performed in $O(\log\frac{m}{sh_u})$ as shown above. For the second part, we perform the reconstruction only when the number of hits on objects marked for deletion $m - M$ exceeds the number of hits on all objects $m$, and, thus, $M \leq \frac{m}{2}$. The reconstruction is performed in $O(M) = O(m)$ time as explained in \emph{Efficient Rebuild} part. Thus we can amortize this $O(m)$ to hits operations performed on logically deleted items. Since there were $O(m - M) = O(m)$ such operations, the amortization ``increases'' their complexities only on some constant and only once, since after the reconstruction the corresponding objects are going to be deleted physically.
\end{proof}

\begin{remark}
For example, if all our operations were successful \texttt{contains}, then the asymptotics for \texttt{contains($u$)} will be $O(\log \frac{m}{sh_u})$ where $m$ is the total number of operations performed.

Furthermore, under the same load we can prove the static optimality property~\cite{knuth1997art}. Let $m_i \leq m$ be the total number of operations when we executed $i$-th operation on $u$, then the total time spent is $O\left(\sum\limits_{i=1}^{sh_u} \log \frac{m_i}{i}\right) = O\left(\sum\limits_{i=1}^{sh_u} \log \frac{m}{i}\right)$ which by Lemma 3 from \cite{CBTree} is equal to $O(sh_i + sh_i \cdot \log \frac{m}{sh_i})$. This is exactly the static optimality property.
\end{remark}

\section{Relaxed Rebalancing}
\label{sec:relaxation-theory}

If we build the straightforward concurrent implementation on top of the sequential implementation described in the previous section, it will obviously suffer in terms of performance since each operation (either \texttt{contains}, \texttt{insert} or \texttt{delete}) must take locks on the whole path to update hits counters. This is not a reasonable approach, especially in the case of the frequent \texttt{contains} operation. 
Luckily for us, \texttt{contains} can be split into two phases: the \emph{search} phase, which traverses the \listname{} and is lock-free, and the \emph{balancing} phase, which updates the counters and maintains ascent and descent conditions. 

A straightforward heuristic is to perform rebalancing infrequently---for example, only once in $c$ operations. 
For this, we propose that the operation perform the update of the global operation counter $m$ and per-object hits counter $sh_u$ only with a fixed probability $1 / c$. 
Conveniently, if the operation does not perform the global operation counter update and the balancing, the counters will not change and, so, all the conditions will still be satisfied. 
The only remaining question is how much this relaxation will affect the data structure's guarantees. The next result characterizes the effects of this relaxation. 

\begin{theorem}
Fix a parameter $c \geq 1$. 
In the relaxed sequential algorithm where operation updates hits counters and performs balancing with probability $\frac{1}{c}$, the hit-operation takes $O\left(c \cdot \log \frac{m}{sh_u}\right)$ expected amortized time, where $m$ is the total number of hit-operations performed on all objects in \listname{} up to the current point in the execution.
\end{theorem}
\begin{proof}
The theoretical analysis above (Theorems~\ref{theorem:total-length} and \ref{thm:main-asymptotics}) is based on the assumption that the algorithm maintains exact values of the counters $m$ and $sh_u$~--- the total number of hit-operations performed to the existing objects and the current number of hit-operations to $u$. However, given the relaxation, the algorithm can no longer rely on $m$ and $sh_u$ since they are now updated only with probability $c$. We denote by $m'$ and $sh'_u$ the relaxed versions of the real counters $m$ and $sh_u$.

The proof consists of two parts. First, we show that the amortized complexity of hits operation to $u$ is equal to $O\left(c \cdot \log \frac{m'}{sh'_u}\right)$ in expectation. Secondly, we show that the approximate counters behave well, i.e., $\EE \left[ \log \frac{m'}{sh'_u} \right]= O\left(\log \frac{m}{sh_u}\right)$. Bringing these two together yields that the amortized complexity of hits operations is $O\left(c \cdot \log \frac{m}{sh_u}\right)$ in expectation.

The first part is proven similarly to Theorem~\ref{thm:main-asymptotics}. We start with the statement that follows from Theorem~\ref{theorem:total-length}: the complexity of any contains operation is equal to $2d + 8y$ where $d$ is the number of objects satisfying the descent condition and $y = 3 + \log \frac{m'}{sh'_u}$. Obviously, we cannot use the same argument as in Theorem~\ref{thm:main-asymptotics} since now $d$ is not equal to the number of descents: the objects which satisfy the descent condition are descended only with probability $\frac{1}{c}$. Thus, we have to bound the sum of $d$ by the total number of descents.

Consider some object $x$ that satisfies the descent condition, i.e. it is counted in $d$ term of the complexity. Then $x$ will either be descended, or will not satisfy the descent condition after $c$ operations passing through it in expectation. 
Mathematically, the event that $x$ is descended follows an exponential distribution with success (demotion) probability $\frac{1}{c}$. Hence, the expected number of operations before $x$ descends is $c$. 

This means that the object $x$ will be counted in terms of type $d$ no more than $c$ times in expectation. By that, the total complexity of all operations is equal to the sum of $8y$ terms plus $2c$ times the number of descents. Since the number of descents cannot exceed the number of ascents, which in turn cannot exceed the sum of the $y$ terms, the total complexity does not exceed the sum of $10 \cdot c \cdot y$ terms. Finally, this means that the amortized complexity complexity of hits operation is $O(c \cdot y) = O\left(c \cdot \log \frac{m}{sh'_u}\right)$ in expectation.

Next, we prove the second main claim, i.e., that  
$$\EE \left( \log \frac{m'}{sh'_u} \right) = O\left(\log \frac{m}{sh_u}\right).$$ 
Note that the relaxed counters $m'$ and $sh'_u$ are Binomial random variables with probability parameter $p = \frac{1}{c}$, and number of trials $m$ and $sh_u$, respectively.

To avoid issues with taking the logarithm of zero, let us bound $\EE\left(\log \frac{m' + 1}{sh'_u + 1}\right)$, which induces only a constant offset. We have: 

\begin{align*}
\EE \left[ \log \frac{m' + 1}{sh'_u + 1} \right] = & \, \EE  \left[ \log (m' + 1) \right]  \,-\,\EE \left[ \log (sh'_u + 1) \right] \\
\underset{\text{Jensen}}\leq & \log (\EE m' + 1)\,-\,\EE \log (sh'_u + 1) 
=  \log (mp + 1)\,-\,\EE \log (sh'_u + 1).
\end{align*}

The next step in our argument will be to lower bound $\EE \log (sh'_u + 1)$. 
For this, we can use the observation that $sh'_u \sim Bin_{sh_u, p}$, the Chernoff bound, and a careful derivation to obtain the following result, whose proof is left to the Appendix~\ref{app:bound}.

\begin{claim}
\label{claim:E:bound}
If $X \sim Bin_{n, p}$ and $np \geq 3 n^{2/3}$ then
$\EE \left[ \log(X + 1) \right] \geq \log np - 4.$
\end{claim}

Based on this, we obtain 
$\log(mp + 1) - \EE [ \log(sh'_u + 1) ]  \leq \log(mp + 1) - \log (sh_u \cdot p) + 4 \leq \log \frac{m}{sh_u} + 5$.

However, this bound works only for the case when $sh_u \cdot p \geq 3 \cdot (sh_u)^{2/3}$. Consider the opposite: $sh_u \leq \frac{27}{p^3}$. Then, $\EE [ \log (sh'_u + 1) ] \geq 0 \geq \log sh_u - \log \frac{27}{p^3}$. Note that the last term is constant, so we can conclude that $\EE[\log \frac{m' + 1}{sh'_u +1}] \leq \log \frac{m}{sh_u} + C$. This matches our initial claim that $\EE[\log \frac{m' + 1}{sh'_u + 1}] = O(\log \frac{m}{sh_u})$. 
\end{proof}


\section{The Concurrent \Listname{}}
\label{sec:practical}

\myparagraph{Overview.} In this section we describe on how to implement scalable lock-based implementation of the \listname{} described in the previous section. The first idea that comes to the mind is to implement the operations as in Lazy Skip-list~\cite{HSBook}: we traverse the data structure in a lock-free manner in the search of $x$ and fill the array of predecessors of $x$ on each level; if $x$ is not found then the operation stops; otherwise, we try to lock all the stored predecessors; if some of them are no longer the predecessors of $x$ we find the real ones or, if not possible, we restart the operation; when all the predecessors are locked we can traverse and modify the backwards path using the presented sequential algorithm without being interleaved. When the total number of operations $m$ becomes a power of two, we have to increase the height of the \listname{} by one: in a straightforward manner, we have to take the lock on the whole data structure and then rebuild it.

There are several major issues with the straightforward implementation described above. At first, the \emph{balancing} part of the operation is too coarse-grained---there are a lot of locks to be taken and, for example, the lock on the topmost level forces the operations to serialize. The second is that the list expansion by freezing the data structure and the following rebuild when $m$ exceeds some power of two is very costly.

\myparagraph{Relaxed and Forward Rebalancing.} 
The first problem can be fixed in two steps. The most important one is to relax guarantees and perform \emph{rebalancing} only periodically, for example, with probability $\frac{1}{c}$ for each operation. 
Of course, this relaxation will affect the bounds---please see  Section~\ref{sec:relaxation-theory} for the proofs. 
However, this relaxation is not sufficient, since we cannot relax the balancing phase of \texttt{insert}($u$) which physically links an object. All these \texttt{insert} functions are going to be serialized due to the lock on the topmost level. Note that without further improvements we cannot avoid taking locks on each predecessor of $x$, since we have to update their counters.
We would like to have more fine-grained implementation.
However, our current sequential algorithm does not allow this, since it updates the path only backwards and, thus, needs the whole path to be locked. 
To address this issue, we introduce a different variant of our algorithm, which does rebalancing \emph{on the forward traversal}.

We briefly describe how this \emph{forward-pass algorithm} works. We maintain the basic structure of the algorithm. Assume we traverse the \listname{} in the search of $x$, and suppose that we are now at the last node $v$ on the level $h$ which precedes $x$. The only node on level $h - 1$ which can be ascended is $v$'s successor on that level, node $u$: 
we check the ascent condition on $u$ or, in other words, compare $\sum_{w \in S_u} hits(C^{h - 1}_w) = hits^h_v - hits^{h - 1}_v$ with $\frac{m}{2^{k - h}}$, and promote $u$, if necessary. 
Then, we iterate through all the nodes on the level $h - 1$ while the keys are less than $x$: if the node satisfies the descent condition, we demote it. Note that the complexity bounds for that algorithm are the same as for the previous one and can be proven exactly the same way (see Theorem~\ref{thm:main-asymptotics}).

The main improvement brought by this forward-pass algorithm is that now the locks can be taken in a hand-over-hand manner: take a lock on the highest level $h$ and update everything on level $h - 1$; take a lock on level $h - 1$, release the lock on level $h$ and update everything on level $h - 2$; take a lock on level $h - 2$, release the lock on level $h - 1$ and update everything on level $h - 3$; and so on. By this locking pattern, the balancing part of different operations is performed in a sequential manner: an operation cannot overtake the previous one and, thus, the $hits$ counters cannot be updated asynchronously. However, at the same time we reduce contention: locks are not taken for the whole duration of the operation.

\myparagraph{Lazy Expansion.} The expansion issue is resolved in a lazy manner.
The \listname{} maintains the counter $zeroLevel$ which represents the current lowest level.
When $m$ reaches the next power of two, $zeroLevel$ is decremented, i.e., we need one more level. (To be more precise, we decrement $zeroLevel$ also lazily: we do this only when some node is going to be demoted from the current lowest level.)
Each node is allocated with an array of $next$ pointers with length $64$ (as discussed, the height $64$ allows us to perform $2^{64}$ operations which is more than enough) and maintains the lowest level to which the node belonged during the last traverse.
When we traverse a node and it appears to have the lowest level higher than $zeroLevel$, we update its lowest level and fill the necessary cells of $next$ pointers.
By doing that we make a lazy expansion of \listname{} and we do not have to freeze whole data structure to rebuild.
For the pseudo-code of lazy expansion, please see Figure~\ref{fig:auxiliary}.
For the pseudo-code of the \listname{}, we refer to Appendix~\ref{sec:code}.

The following Theorem trivially holds due to the specificity of skip-list: if an operation reaches  a sub-list of lower height than its target elementm it will still find it, if it is present.
\begin{theorem}
The presented concurrent \listname{} algorithm is linearizable.
\end{theorem}

\vspace{-0.5cm}
\section{Experimental Evaluation}
\myparagraph{Environment and Methodology.} We evaluate algorithms on a 4-socket Intel Xeon Gold 6150 2.7 GHz server with 18 threads per socket. The code is written in C++ and was compiled by MinGW GCC 6.3.0 compiler with \texttt{-O2} optimizations. 
Each experiment was performed $10$ times and all the values presented are averages. 
The code is available at \url{https://cutt.ly/disc2020353}.

\myparagraph{Workloads and Parameters.} Due to space constraints, our experiments in this section consider read-only workloads with unbalanced access distribution, which are the focus of our paper. We also execute uniform and read-write workloads, whose results we present in Appendix~\ref{sec:additional-results}. In our experiments, we describe a family of workloads by $n-x-y$, which should be read as: given $n$ keys, $x\%$ of the \texttt{contains} are performed on $y\%$ of the keys.
More precisely, we first populate the \listname{} with $n$ keys and randomly choose a set of ``popular'' keys $S$ of size $y \cdot n$. 
We then start $T$ threads, each of which iteratively picks an element and performs the \texttt{contains} operation, for $10$ seconds. With probability $x$ we choose a random element from $S$, otherwise, we choose an element outside of $S$ uniformly at random.

For our experiments, we choose the following workloads: $10^5-90-10$, $10^5-95-5$ and $10^5-99-1$. That is, $90\%$, $95\%$, and $99\%$ of the operations go into $10\%$, $5\%$, and $1\%$ of the keys, respectively. 
Further, we vary the \emph{balancing rate/probability}, which we denote by $p$: this is the probability that a given operation will update hit counters and perform rebalancing. 
In Appendix~\ref{sec:additional-results}, we also examine uniform and Zipf distributions. 

\begin{table*}
\resizebox{\linewidth}{!}{%
\begin{tabular}{|c|c|c|c|c|c|c|c|c|}
\hline
$10^5-90-10$ & Skip-list & SL $p=1$ & SL $p=\frac{1}{2}$ & SL $p=\frac{1}{5}$ & SL $p=\frac{1}{10}$ & SL $p=\frac{1}{100}$ & SL $p=\frac{1}{1000}$\\\hline
ops/sec & 2874600.0 & 0.60x & 0.78x & 1.00x & 1.10x & 1.12x & 1.02x \\\hline
length & 30.81 & 23.06 & 23.07 & 23.08 & 23.13 & 23.75 & 25.06 \\\hline
& & CBTree $p=1$ & CBTree $p=\frac{1}{2}$ & CBTree $p=\frac{1}{5}$ & CBTree $p=\frac{1}{10}$ & CBTree $p=\frac{1}{100}$ & CBTree $p=\frac{1}{1000}$ \\\hline
ops/secs & & 1.15x & 1.36x & 1.59x & 1.71x & 1.71x & 1.52x\\\hline
length & & 9.13 & 9.14 & 9.15 & 9.17 & 9.37 & 9.81\\\hline
\end{tabular}
}
\label{table:90-10}
\caption{Operations per second and average length of a path on $10^5-90-10$ workload.}

\resizebox{\linewidth}{!}{%
\begin{tabular}{|c|c|c|c|c|c|c|c|c|}
\hline
$10^5-95-5$ & Skip-list & SL $p=1$ & SL $p=\frac{1}{2}$ & SL $p=\frac{1}{5}$ & SL $p=\frac{1}{10}$ & SL $p=\frac{1}{100}$ & SL $p=\frac{1}{1000}$\\\hline
ops/sec & 2844520.0 & 0.69x & 0.93x & 1.21x & 1.34x & 1.39x & 1.17x \\\hline
length & 30.84 & 21.62 & 21.63 & 21.65 & 21.70 & 22.33 & 24.46 \\\hline
& & CBTree $p=1$ & CBTree $p=\frac{1}{2}$ & CBTree $p=\frac{1}{5}$ & CBTree $p=\frac{1}{10}$ & CBTree $p=\frac{1}{100}$ & CBTree $p=\frac{1}{1000}$ \\\hline
ops/secs & & 1.33x & 1.61x & 1.90x & 2.04x & 2.09x & 1.79x\\\hline
length & & 8.61 & 8.61 & 8.62 & 8.65 & 8.90 & 9.58\\\hline
\end{tabular}
}
\label{table:95-5}
\caption{Operations per second and average length of a path on $10^5-95-5$ workload.}

\resizebox{\linewidth}{!}{%
\begin{tabular}{|c|c|c|c|c|c|c|c|c|}
\hline
$10^5-99-1$ & Skip-list & SL $p=1$ & SL $p=\frac{1}{2}$ & SL $p=\frac{1}{5}$ & SL $p=\frac{1}{10}$ & SL $p=\frac{1}{100}$ & SL $p=\frac{1}{1000}$\\\hline
ops/sec & 3559320.0 & 0.85x & 1.19x & 1.65x & 1.89x & 2.01x & 1.64x \\\hline
length & 31.00 & 17.13 & 17.16 & 17.23 & 17.30 & 18.59 & 21.00 \\\hline
& & CBTree $p=1$ & CBTree $p=\frac{1}{2}$ & CBTree $p=\frac{1}{5}$ & CBTree $p=\frac{1}{10}$ & CBTree $p=\frac{1}{100}$ & CBTree $p=\frac{1}{1000}$ \\\hline
ops/secs & & 1.37x & 1.72x & 2.06x & 2.25x & 2.36x & 2.04x \\\hline
length & & 7.25 & 7.23 & 7.26 & 7.28 & 7.52 & 8.53\\\hline
\end{tabular}
}
\label{table:99-1}
\caption{Operations per second and average length of a path on $10^5-99-1$ workload.}
\end{table*}

\myparagraph{Goals and Baselines.}
We aim to determine whether 1) the \listname{} can improve over the throughput of the baseline skip-list by successfully leveraging the skewed access distribution; 
2) whether it scales, and what is the impact of update rates and number of threads; and, finally,
3) whether it can be competitive with the CBTree data structure in sequential and concurrent scenarios.

\myparagraph{Sequential evaluation.} In the first round of experiments, we compare how the single-threaded \listname{} performs under the chosen workloads. We execute it with different settings of $p$, the probability of adjustment, taking values $1$, $\frac{1}{2}$, $\frac{1}{5}$, $\frac{1}{10}$, $\frac{1}{100}$ and  $\frac{1}{1000}$. 
We compare against the sequential skip-list and CB-Tree. We measure two values: the number of operations per second and the average length of the path traversed. The results are presented in Tables 1---3 (\Listname{} is abbreviated SL). For readability, throughput results are presented relative to the skip-list baseline.  

Relative to the skip-list, the first observation is that, for high update rates (1 through 1/5), the \listname{} predictably only matches or even loses performance. 
However, this trend improves as we reduce the update rate, 
and, more significantly, as we increase the access rate imbalance: for $99-1$, the sequential \listname{} obtains a throughput improvement of $2\times$. This improvement directly correlates with the length of the access path (see third row). 
At the same time, notice the negative impact of very low update rates (last column), as the average path length increases, which leads to higher average latency and decreased throughput. 
We empirically found the best update rate to be around $1 / 100$, trading off latency with per-operation cost.

Relative to the sequential CBTree, we notice that the \listname{} generally yields lower throughput. This is due to two factors: 1) the CBTree is able to yield shorter access paths, due to its structure and constants; 2) the tree tends to have better cache behavior relative to the skip-list backbone. 
Given the large difference in terms of average path length, it may seem surprising that the \listname{} is able to provide close performance. This is because of the caching mechanism: as long as the path length for popular elements is short enough so that they all are mostly in cache, the average path length is not critical. We will revisit this observation in the concurrent case.

\begin{figure*}
\begin{subfigure}{.33\textwidth}
\includegraphics[width=\linewidth]{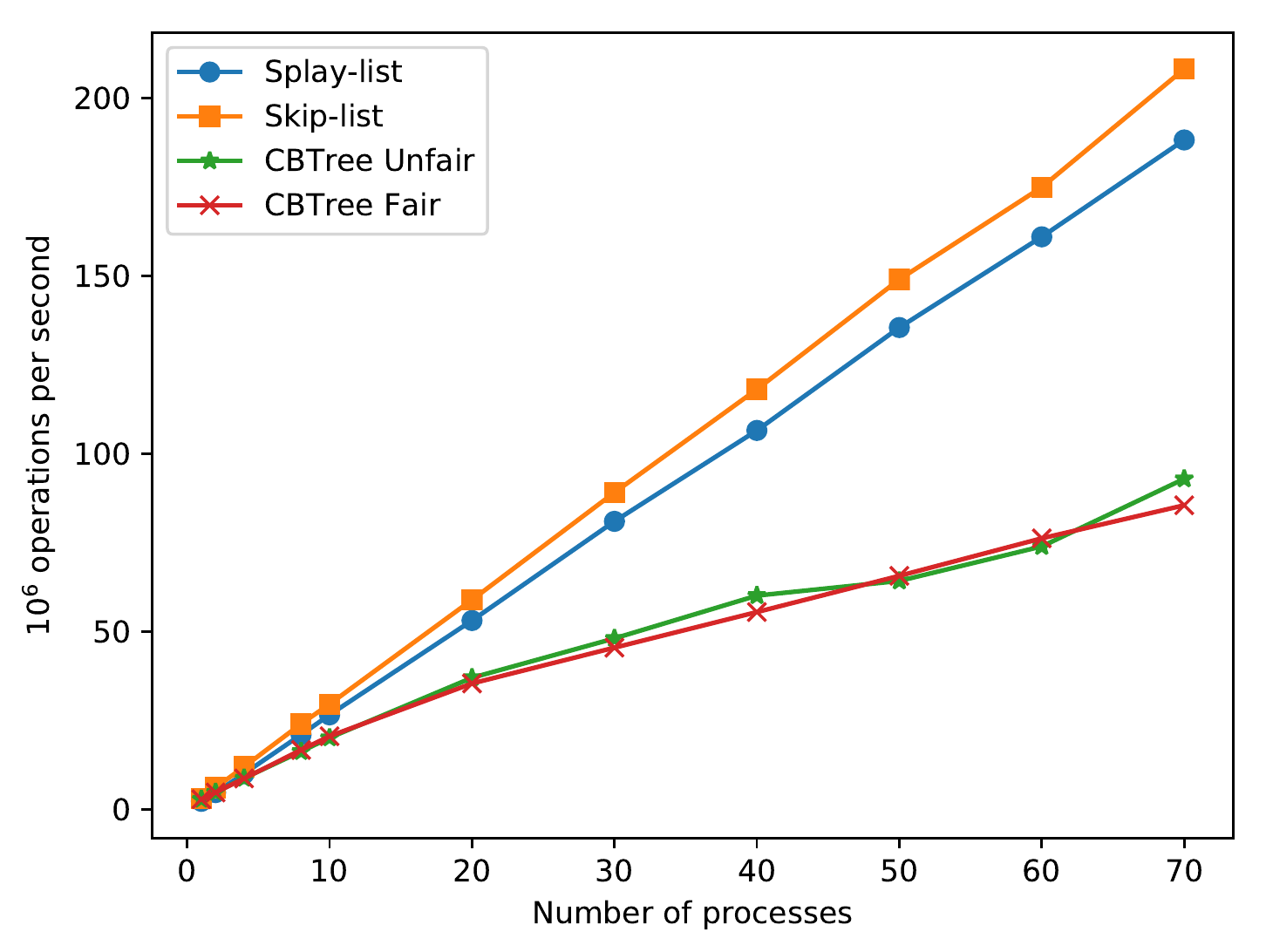}
\caption{$p = \nicefrac{1}{10}$}
\end{subfigure}%
\begin{subfigure}{.33\textwidth}
\includegraphics[width=\linewidth]{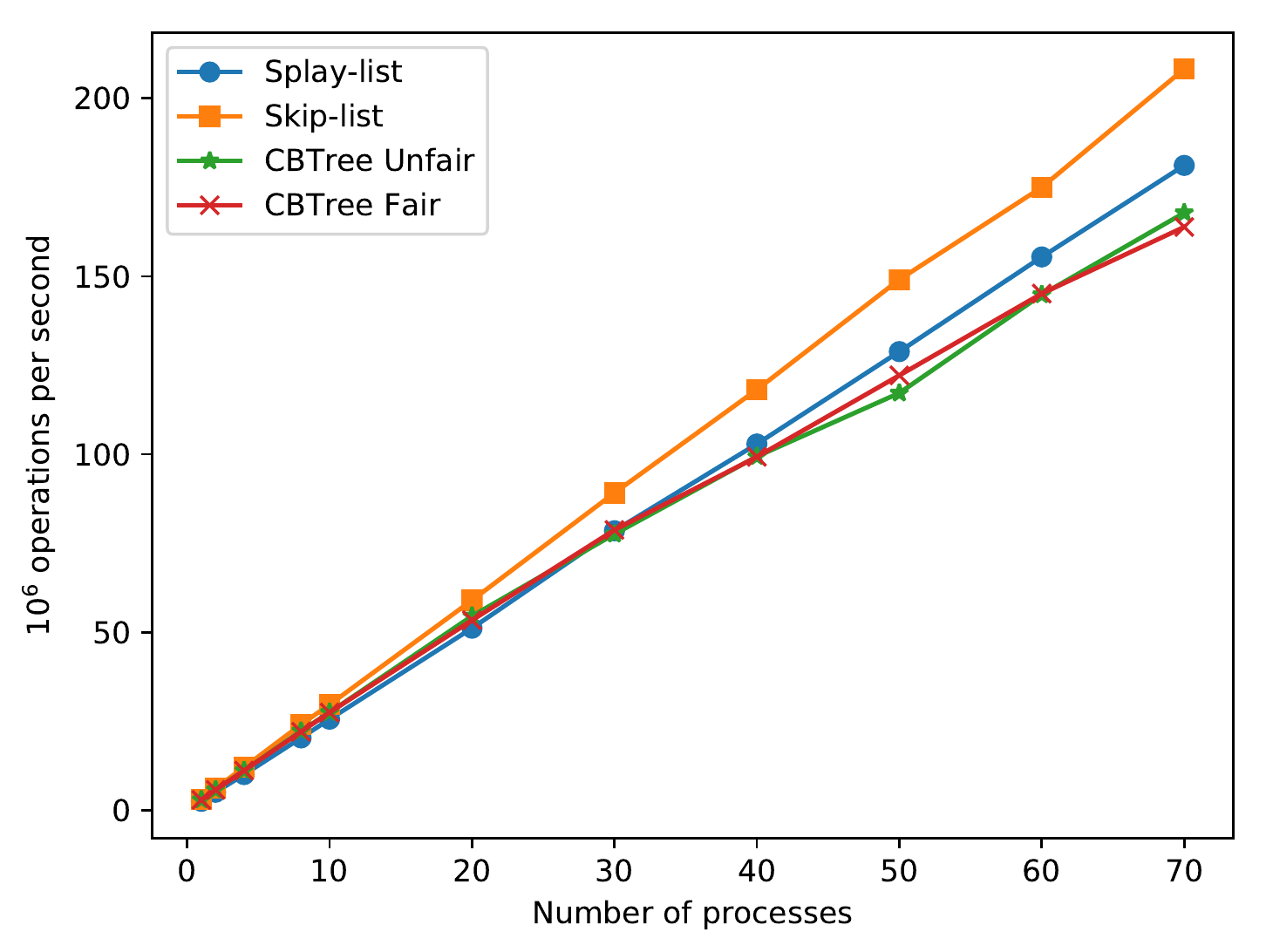}
\caption{$p = \nicefrac{1}{100}$}
\end{subfigure}%
\begin{subfigure}{.33\textwidth}
\includegraphics[width=\linewidth]{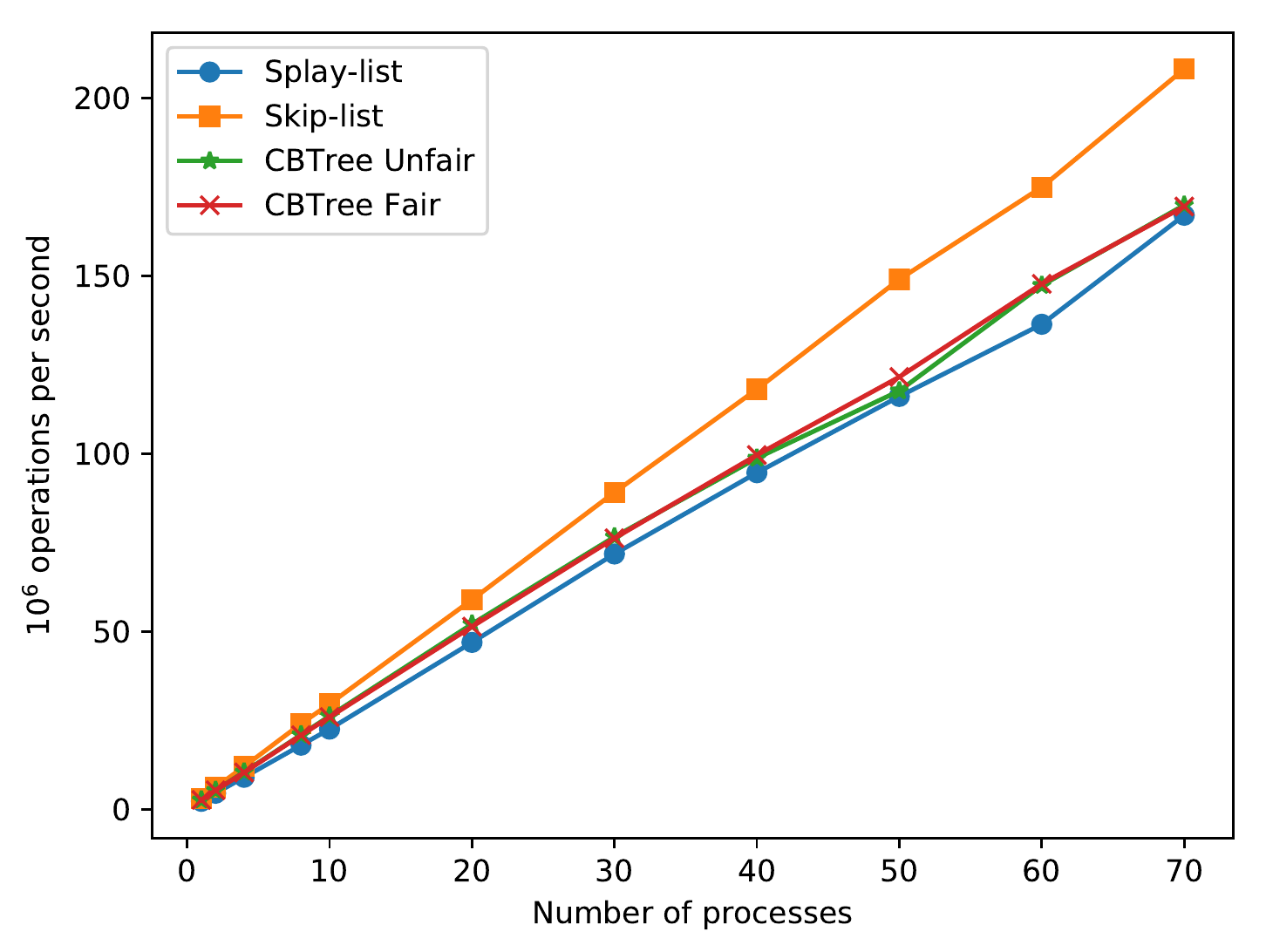}
\caption{$p = \nicefrac{1}{1000}$}
\end{subfigure}
\caption{Concurrent throughput for $10^5-90-10$ workload.}
\label{fig:90-10}
\end{figure*}

\begin{figure*}
\begin{subfigure}{.33\textwidth}
\includegraphics[width=\linewidth]{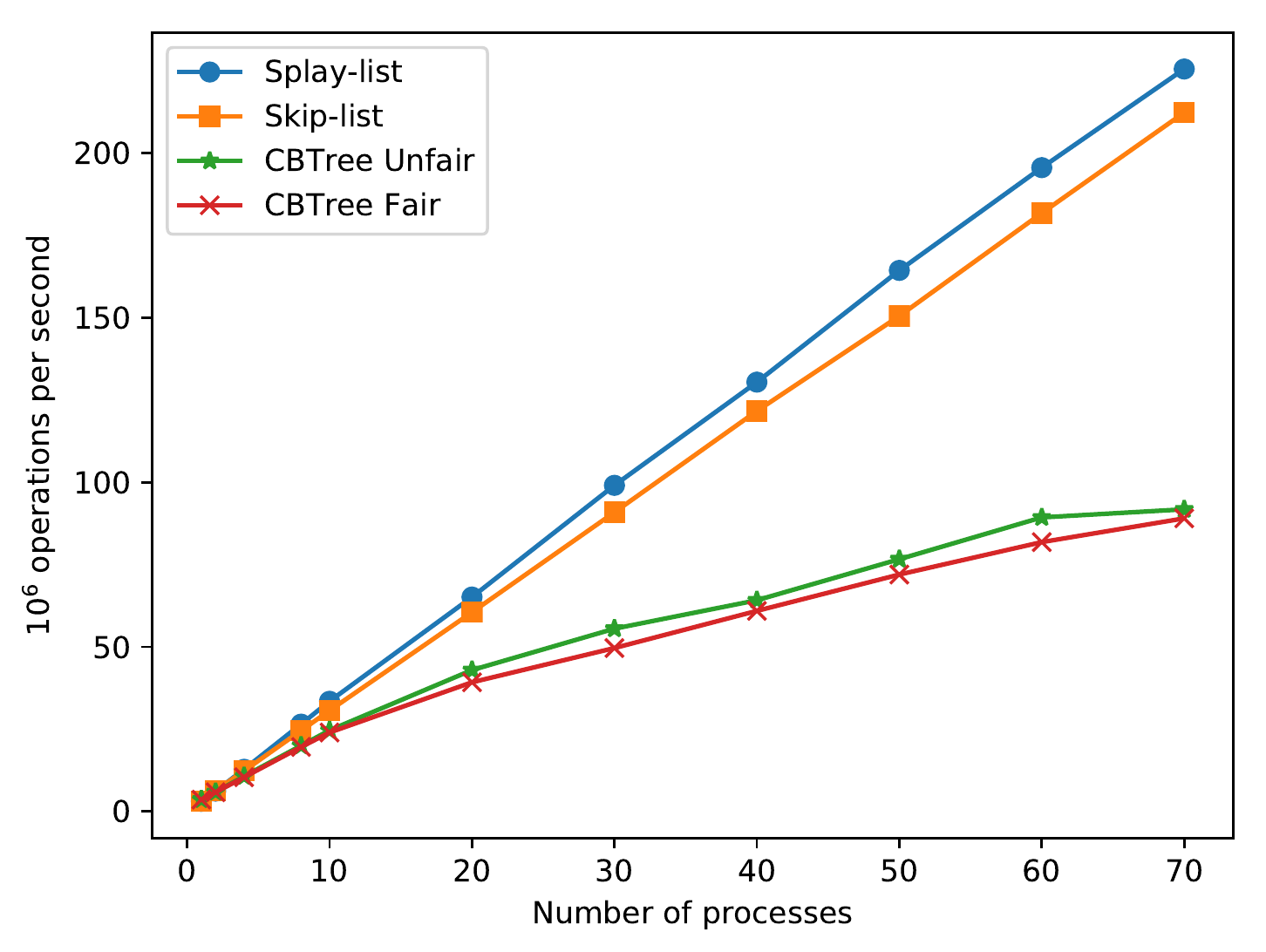}
\caption{$p = \nicefrac{1}{10}$}
\end{subfigure}%
\begin{subfigure}{.33\textwidth}
\includegraphics[width=\linewidth]{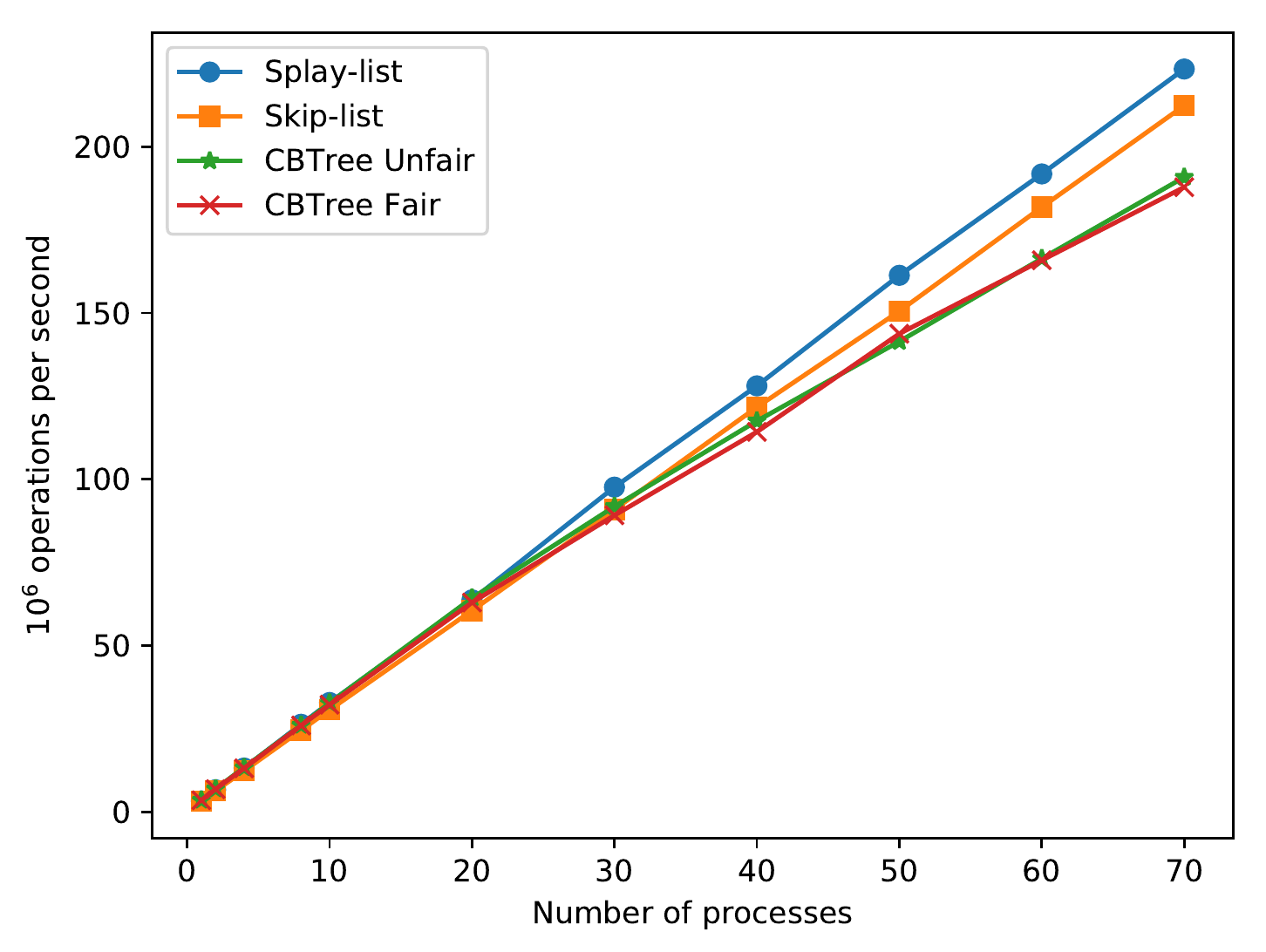}
\caption{$p = \nicefrac{1}{100}$}
\end{subfigure}%
\begin{subfigure}{.33\textwidth}
\includegraphics[width=\linewidth]{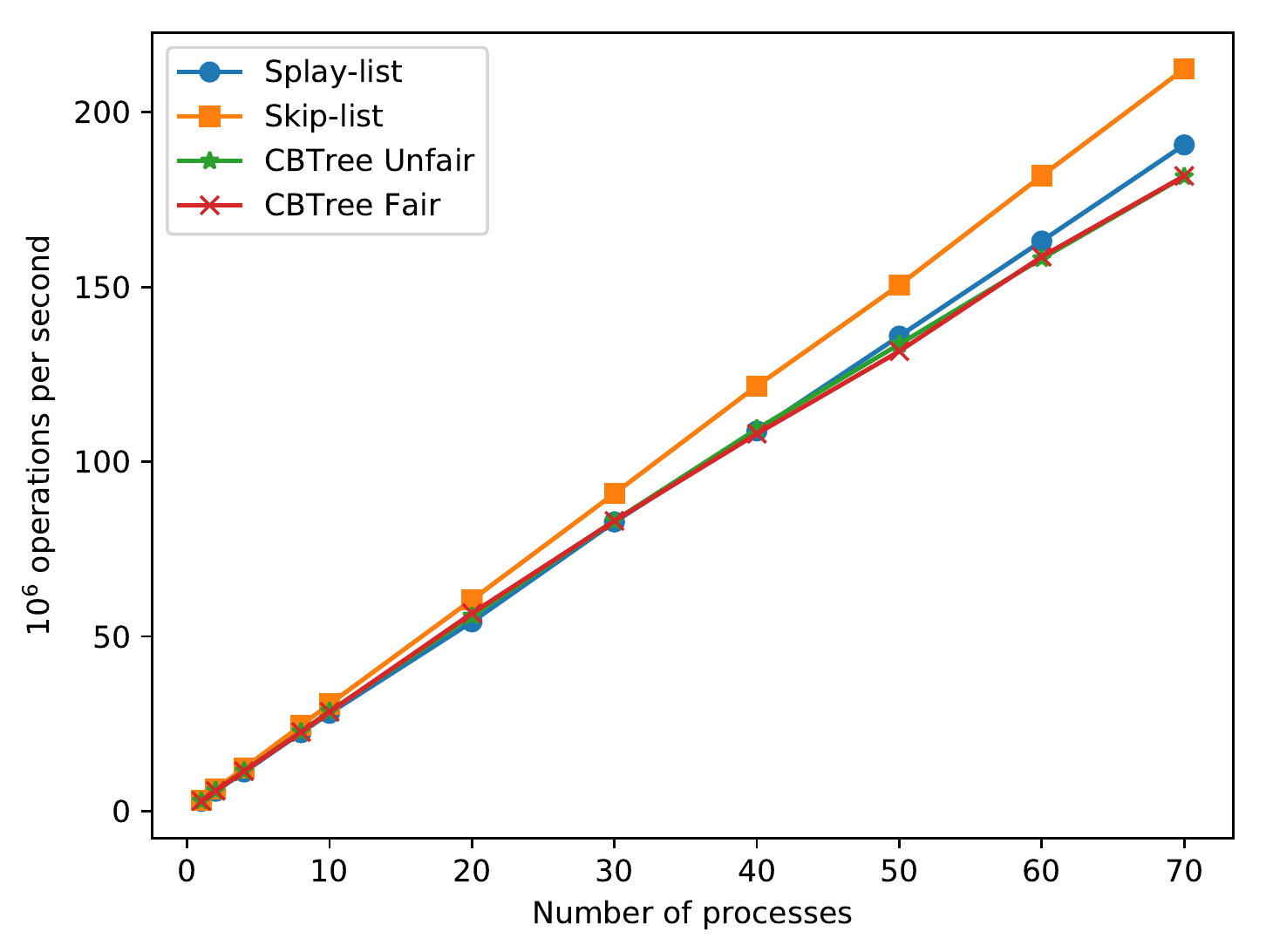}
\caption{$p = \nicefrac{1}{1000}$}
\end{subfigure}
\caption{Concurrent throughput for  $10^5-95-5$ workload.}
\label{fig:95-5}
\end{figure*}

\begin{figure*}
\begin{subfigure}{.33\textwidth}
\includegraphics[width=\linewidth]{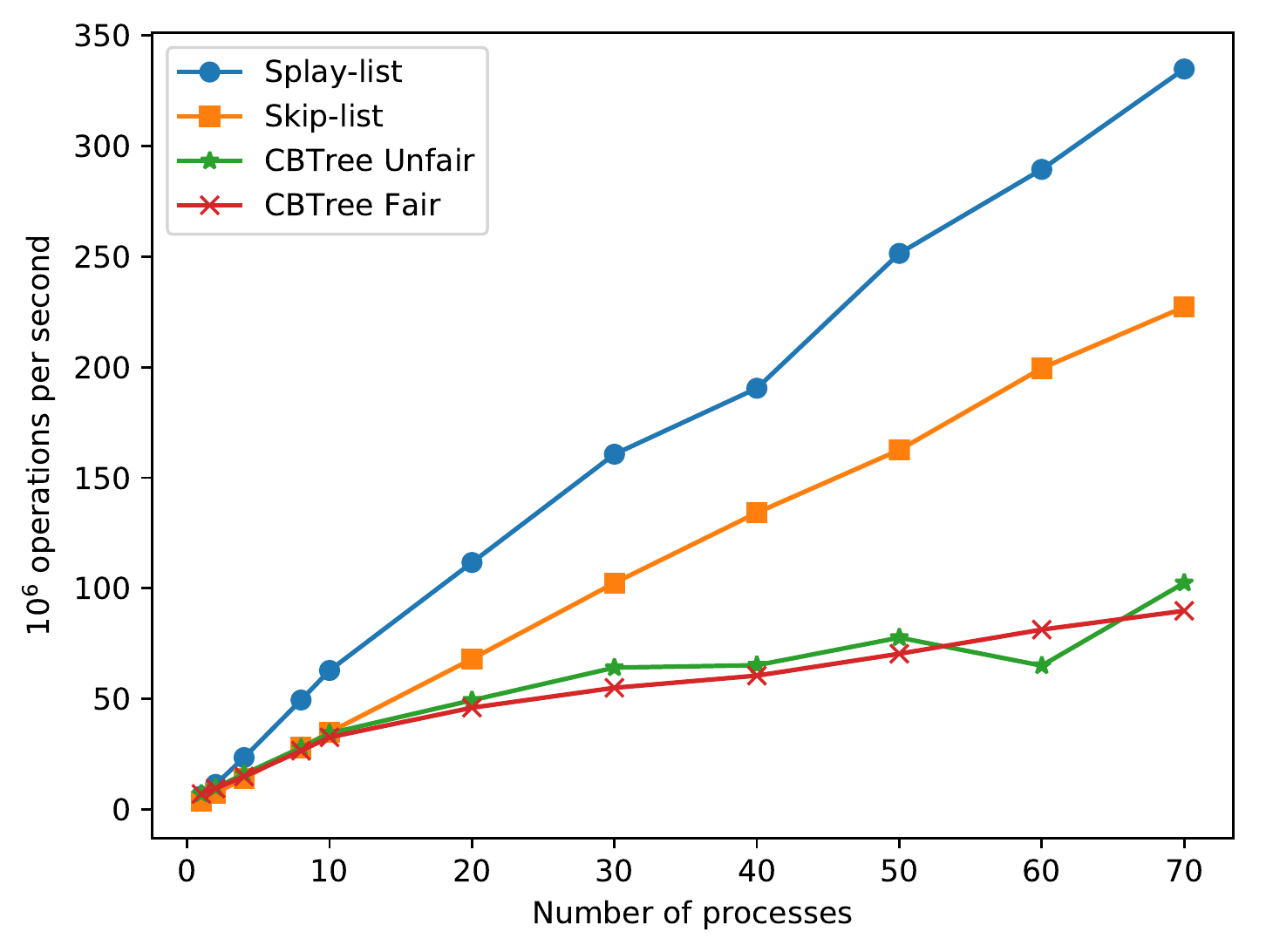}
\caption{$p = \nicefrac{1}{10}$}
\end{subfigure}
\begin{subfigure}{.33\textwidth}
\includegraphics[width=\linewidth]{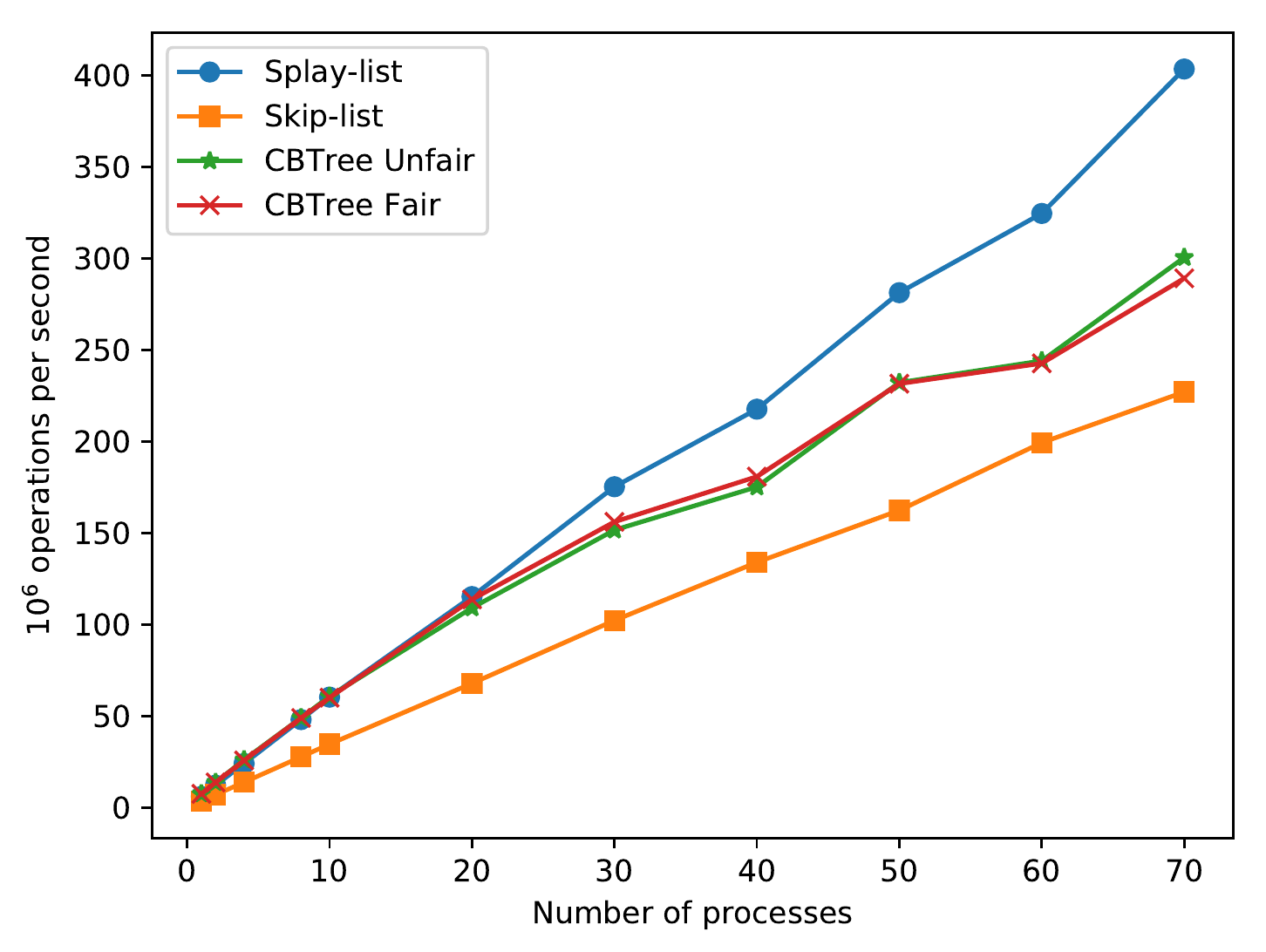}
\caption{$p = \nicefrac{1}{100}$}
\end{subfigure}%
\begin{subfigure}{.33\textwidth}
\includegraphics[width=\linewidth]{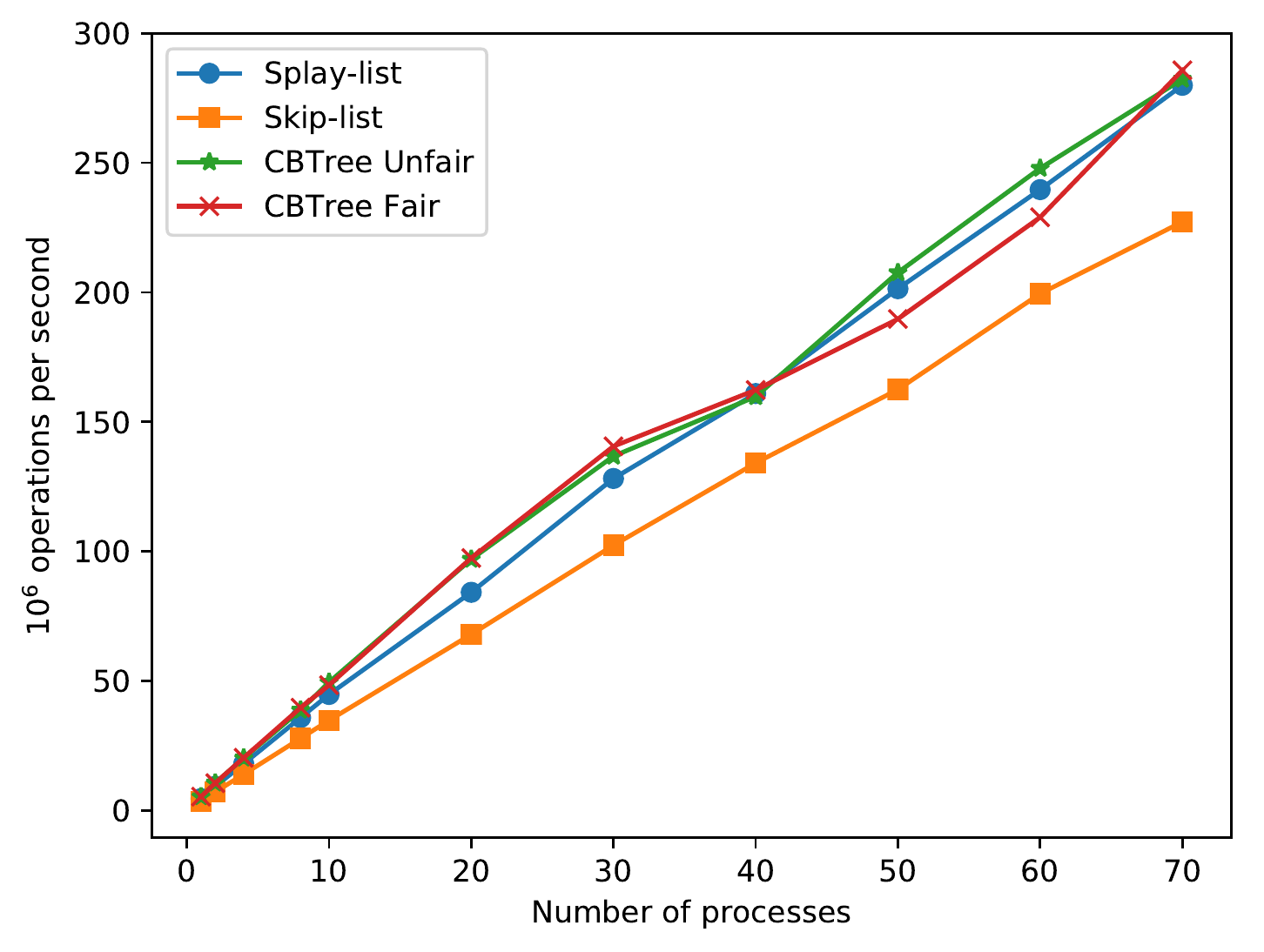}
\caption{$p = \nicefrac{1}{1000}$}
\end{subfigure}
\caption{Concurrent throughput for $10^5-99-1$ workload.}
\label{fig:99-1}
\vspace{-0.2cm}
\end{figure*}

\myparagraph{Concurrent evaluation.}
Next, we analyze concurrent performance.  
Unfortunately, the original implementation of the CBTree is not available, and we therefore re-implemented it in our framework. 
Here, we make an important distinction relative to usage: the authors of the CBTree paper propose to use a single thread to perform all the rebalancing. However, this approach is not standard, as in practice, updates could come at different threads. Therefore, we implement two versions of the CBTree, one in which updates are performed by a single thread (CBTree-Unfair), and one in which updates can be performed by every thread (CBTree-Fair). 
In both cases, synchronization between readers and writers is performed via an efficient readers-writers lock~\cite{rwlock}, which prevents concurrent updates to the tree. We note that in theory we could further optimize the CBTree to allow fully-concurrent updates via fine-grained synchronization. However, 1) this would require a significant re-working of their algorithm; 2)
as we will see below, this would not change  results significantly. 

Our experiments, presented in Figures~\ref{fig:90-10}, \ref{fig:95-5}, and \ref{fig:99-1}, analyze the performance of the \listname{} relative to standard skip-list and the CBTree across different workloads (one per figure), different update rates (one per panel), and thread counts (X axis).

Examining the figures, first notice the relatively good scalability of the \listname{} under all chosen update rates and workloads. By contrast, the CBTree scales well for moderately skewed workloads and low update rates, but performance decays for skewed workloads and high update rates (see for instance Figure \ref{fig:99-1}(a)). 
We note that, in the former case the CBTree matches the performance of the \listname{} in the low-update case (see Figure \ref{fig:90-10}(c)), but its performance can decrease significantly if the update rates are reasonably high ($p = \nicefrac{1}{100}$). 
We further note the limited impact of whether we consider the fair or unfair variant of the CBTree (although the Unfair variant usually performs  better).

These results may appear surprising given that the \listname{} generally has longer access paths. However, it benefits significantly from the fact that it allows additional concurrency, and that the caching mechanism serves to hide some of its additional access cost.
Our intuition here is that one critical measure is which fraction of the ``popular'' part of the data structure fits into the cache. 
This suggests that the \listname{} can be practically competitive relative to the CBTree on a subset of workloads.

\myparagraph{Additional Experiments.}
The experiments in Appendix~\ref{sec:additional-results} examine 1) the overheads in the uniform access case, 2) performance for a Zipf access distribution; 3) performance under moderate insert/delete rates. We also examine performance over longer runs, as well as the correlation between element height in the list and its ``popularity.''

\vspace{-1em}
\section{Discussion}
\vspace{-0.5em}

We revisited the question of efficient self-adjusting concurrent data structures, and presented the first instance of a self-adjusting concurrent skip-list, addressing an open problem posed by~\cite{CBTree}.
Our design ensures static optimality, and has an arguably simple structure and implementation, which allows for additional concurrency and good performance under skewed access. 
In addition, it is the first design to provide guarantees under approximate access counts, required for good practical behavior. 
In future work, we plan to expand the experimental evaluation to include a range of real-world workloads, and to prove the guarantees under concurrent access.

\newpage

\bibliographystyle{plain}
\bibliography{ref}

\begin{thebibliography}{10}

\bibitem{CBTree}
Yehuda Afek, Haim Kaplan, Boris Korenfeld, Adam Morrison, and Robert~E. Tarjan.
\newblock Cbtree: A practical concurrent self-adjusting search tree.
\newblock In {\em Proceedings of the 26th International Conference on
  Distributed Computing}, DISC'12, pages 1--15, Berlin, Heidelberg, 2012.
  Springer-Verlag.

\bibitem{215945}
Maya Arbel-Raviv, Trevor Brown, and Adam Morrison.
\newblock Getting to the root of concurrent binary search tree performance.
\newblock In {\em 2018 {USENIX} Annual Technical Conference ({USENIX} {ATC}
  18)}, pages 295--306, Boston, MA, July 2018. {USENIX} Association.

\bibitem{bagchi2005biased}
Amitabha Bagchi, Adam~L Buchsbaum, and Michael~T Goodrich.
\newblock Biased skip lists.
\newblock {\em Algorithmica}, 42(1):31--48, 2005.

\bibitem{bose2008dynamic}
Prosenjit Bose, Karim Dou{\"\i}eb, and Stefan Langerman.
\newblock Dynamic optimality for skip lists and b-trees.
\newblock In {\em Proceedings of the nineteenth annual ACM-SIAM symposium on
  Discrete algorithms}, pages 1106--1114, 2008.

\bibitem{brown2017techniques}
Trevor Brown.
\newblock {\em Techniques for Constructing Efficient Data Structures}.
\newblock PhD thesis, PhD thesis, University of Toronto, 2017.

\bibitem{ciriani2002static}
Valentina Ciriani, Paolo Ferragina, Fabrizio Luccio, and Shanmugavelayutham
  Muthukrishnan.
\newblock Static optimality theorem for external memory string access.
\newblock In {\em The 43rd Annual IEEE Symposium on Foundations of Computer
  Science, 2002. Proceedings.}, pages 219--227. IEEE, 2002.

\bibitem{cooper2010benchmarking}
Brian~F Cooper, Adam Silberstein, Erwin Tam, Raghu Ramakrishnan, and Russell
  Sears.
\newblock Benchmarking cloud serving systems with ycsb.
\newblock In {\em Proceedings of the 1st ACM symposium on Cloud computing},
  pages 143--154, 2010.

\bibitem{rwlock}
Andreia Correia and Pedro Ramalhete.
\newblock Scalable reader-writer lock in c++1x.
\newblock
  \url{http://concurrencyfreaks.blogspot.com/2015/01/scalable-reader-writer-lock-in-c1x.html},
  2015.

\bibitem{Ellen10}
Faith Ellen, Panagiota Fatourou, Eric Ruppert, and Franck van Breugel.
\newblock Non-blocking binary search trees.
\newblock In {\em Proceedings of the 29th ACM SIGACT-SIGOPS Symposium on
  Principles of Distributed Computing}, PODC '10, pages 131--140, New York, NY,
  USA, 2010. ACM.

\bibitem{FraserThesis}
Keir Fraser.
\newblock {Practical lock-freedom}.
\newblock Technical Report UCAM-CL-TR-579, University of Cambridge, Computer
  Laboratory, February 2004.

\bibitem{Fraser}
Keir Fraser.
\newblock {\em Practical lock-freedom}.
\newblock PhD thesis, PhD thesis, Cambridge University Computer Laboratory,
  2003. Also available as Technical Report UCAM-CL-TR-579, 2004.

\bibitem{LazySkiplist}
Maurice Herlihy, Yossi Lev, Victor Luchangco, and Nir Shavit.
\newblock A simple optimistic skiplist algorithm.
\newblock In {\em Proceedings of the 14th international conference on
  Structural information and communication complexity}, SIROCCO'07, pages
  124--138, Berlin, Heidelberg, 2007. Springer-Verlag.

\bibitem{HSBook}
Maurice Herlihy and Nir Shavit.
\newblock {\em The Art of Multiprocessor Programming}.
\newblock Morgan Kaufmann Publishers Inc., San Francisco, CA, USA, 2008.

\bibitem{iacono2001alternatives}
John Iacono.
\newblock Alternatives to splay trees with o (log n) worst-case access times.
\newblock In {\em Proceedings of the twelfth annual ACM-SIAM symposium on
  Discrete algorithms}, pages 516--522. Society for Industrial and Applied
  Mathematics, 2001.

\bibitem{knuth1997art}
Donald~Ervin Knuth.
\newblock {\em The art of computer programming}, volume~3.
\newblock Pearson Education, 1997.

\bibitem{DougSkiplist}
Doug Lea, 2007.
\newblock
  \url{http://java.sun.com/javase/6/docs/api/java/util/concurrent/ConcurrentSkipListMap.html}.

\bibitem{martel1991self}
Charles Martel.
\newblock Self-adjusting multi-way search trees.
\newblock {\em Information Processing Letters}, 38(3):135--141, 1991.

\bibitem{Michael}
Maged~M Michael.
\newblock High performance dynamic lock-free hash tables and list-based sets.
\newblock In {\em Proceedings of the fourteenth annual ACM symposium on
  Parallel algorithms and architectures}, pages 73--82. ACM, 2002.

\bibitem{Natarajan14}
Aravind Natarajan and Neeraj Mittal.
\newblock Fast concurrent lock-free binary search trees.
\newblock In {\em Proceedings of the 19th ACM SIGPLAN Symposium on Principles
  and Practice of Parallel Programming}, PPoPP '14, pages 317--328, New York,
  NY, USA, 2014. ACM.

\bibitem{TPCC}
Meikel Poess and Chris Floyd.
\newblock New tpc benchmarks for decision support and web commerce.
\newblock {\em ACM Sigmod Record}, 29(4):64--71, 2000.

\bibitem{Pugh}
William Pugh.
\newblock Concurrent maintenance of skip lists.
\newblock 1998.

\bibitem{sleator1985self}
Daniel~Dominic Sleator and Robert~Endre Tarjan.
\newblock Self-adjusting binary search trees.
\newblock {\em Journal of the ACM (JACM)}, 32(3):652--686, 1985.

\end{thebibliography}

\appendix
\section{Deferred Proofs}
\label{app:bound}

{
\def\thetheorem{\ref{claim:E:bound}}
\begin{claim}
If $X \sim Bin_{n, p}$ and $np \geq 3 n^{2/3}$ then
$$\EE \left[ \log(X + 1) \right] \geq \log np - 4.$$
\end{claim}
\addtocounter{theorem}{-1}
}
\begin{proof}
Recall the standard Chernoff bound, which says that if  $X \sim Bin_{n, p}$, then
$P(|X - np| > \delta np) \leq 2 e^{-\mu \delta^2 / 3}$.
Applying this with $\delta = \frac{1}{n^{1/3}p}$, we obtain $P(|X - np| > n^{\frac{2}{3}}) \leq 2 e^{-\frac{n^{1/3}}{3p^2}}$.

$\EE \log(X + 1) = \EE \log(np + (X - np + 1)) = \log np + \EE \log \left(1 + \frac{X - np + 1}{np}\right) = \log np + \sum\limits_{k = 0}^n p_k \log \left(1 + \frac{k - np + 1}{np}\right) \underset{\substack{\text{Taylor series and } \\ 1 + \frac{k - np + 1}{np} \geq \frac{1}{np}}}\geq \\ \geq \log np + \sum\limits_{k = np - n^{2/3}}^{np + n^{2/3}} p_k \left( \frac{k - np + 1}{np} - \frac{(k - np + 1)^2}{2n^2p^2} + \ldots \right) + P(|X - np| > n^{\frac{2}{3}}) \cdot \log \frac{1}{np} \geq \log np\,- \sum\limits_{k = np - n^{2/3}}^{np + n^{2/3}} p_k \left( \frac{2n^{2/3}}{np} + \frac{(2n^{2/3})^2}{2(np)^2} + \ldots \right) - 2 \log np \cdot e^{-\frac{n^{1/3}}{3p^2}} \underset{\sum_{k = np - n^{2/3}}^{np + n^{2/3}} p_k \leq 1}\geq \log np - \left(\frac{2n^{2/3}}{np} + \frac{(2n^{2/3})^2}{(np)^2}  + \ldots \right) - 2 \log np 
\cdot e^{-\frac{n^{1/3}}{3p^2}} = \log np - \frac{1}{1 - \frac{2n^{2/3}}{np}} - 2  \log np \cdot e^{-\frac{n^{1/3}}{3p^2}} \geq \log np - 3 - 2 \log np \cdot e^{-\frac{n^{1/3}}{3p^2}} \geq \log np - 4$.
\end{proof}

\section{Pseudo-code}
\label{sec:code}

In this section we introduce the pseudo-code for \texttt{contains} operation. \texttt{Insert} and \texttt{delete} (that simply marks) operations are performed similarly. The rebuild is a little bit complicated since we have to freeze whole data structure, however, since we talk about lock-based implementations it can be simply done by providing the global lock on the data structure.

The main class that is used is \texttt{Node} (Figure~\ref{fig:structures}). It contains nine fields: 1)~$key$ field stores the corresponding key, 2)~$value$ field stores the value stored for the corresponding key, 3)~$zeroLevel$ field indicates the lowest sub-list to which the object belongs (for lazy expansion), 4)~$topLevel$ field indicates the topmost sub-list to which the object belongs,
5)~$lock$ field allows to lock the object, 6)~$selfhits$ field stores the total number of hit-operations performed to $key$, i.e., $sh_{key}$, 7)~$next[h]$ is the succesor of the object in the sub-list of height $h$, 8)~$hits[h]$ equals to $hits^h_{key}$ or, in other words, $C^h_{key} - \idl{selfhits}$, and, finally, 9)~$deleted$ mark that indicates whether the key is logically deleted.
The \listname{} itself is represented by class \texttt{SplayList} with five fields: 1)~$m$ field stores the total number of hit-operations, 2)~$M$ field stores the total number of hit-operations to non-marked objects, 3)~$zeroLevel$ indicates the current lowest level (for lazy restructuring), 4)~$head$ and $tail$ are sentinel nodes with $-\infty$ and $+\infty$ keys, correspondingly.
Moreover, the algorithm has a parameter $p$ which is the probability how often we should perform the balancing part of contains function.

\begin{lstlisting}
class Node:
  K key
  V value
  int zeroLevel
  int topLevel
  Lock lock
  int selfhits
  Node next[MAX_LEVEL]
  int hits[MAX_LEVEL]
  bool deleted
  
class SplayList: 
  int m
  int M
  int zeroLevel
  Node head
  Node tail

SplayList list
double p
\end{lstlisting}
\vspace{-0.5cm}
\captionof{figure}{The data structure class definitions.}
\label{fig:structures}

The \texttt{contains} function is depicted at Figure~\ref{fig:contains}. If \texttt{find} did not find an object with the corresponding key then we return \texttt{false}. Otherwise, we execute balancing part, i.e., function \texttt{update}, with the probability $p$.
\begin{lstlisting}
fun contains(K key):
  Node node $\leftarrow$ find(key)
  if node = null:
    return false
  if random() < p:
    update(key)
  return not node.deleted
\end{lstlisting}
\vspace{-0.5cm}
\captionof{figure}{Contains function}
\label{fig:contains}

The \texttt{find} method which checks the existence of the $key$ almost identical to the standard \texttt{find} function in skip-lists. It is presented on the following Figure~\ref{fig:find}.
\begin{lstlisting}
fun find(K key):
  pred $\leftarrow$ list.head
  succ $\leftarrow$ head.next[MAX_LEVEL]
  for level $\leftarrow$ MAX_LEVEL-1 .. zeroLevel:
    updateUpToLevel(pred, level) |\label{line:level:check:1}|
    succ $\leftarrow$ pred.next[level]
    if succ = null:
      continue
    updateUpToLevel(succ, level) |\label{line:level:check:2}|
    while succ.key < key:
      pred $\leftarrow$ succ
      succ $\leftarrow$ pred.next[level]
      if succ = null:
        break
      updateUpToLevel(succ, level)
    if succ $\ne$ null and succ.key = key:
      return succ
  return null
\end{lstlisting}
\vspace{-0.5cm}
\captionof{figure}{Find function}
\label{fig:find}

Note, that as discussed in lazy expansion part, when we pass the object we check (Figure~\ref{fig:find} Lines~\ref{line:level:check:1}~and~\ref{line:level:check:2}) whether it should belong to lower levels, i.e., the expansion was performed, and if it is we update it. For the lazy expansion functions we refer to the next Figure~\ref{fig:auxiliary}.
\begin{lstlisting}
// this function is called only when node.lock is taken
fun updateZeroLevel(Node node):
  if node.zeroLevel > list.zeroLevel:
    node.hits[node.zeroLevel - 1] $\leftarrow$ 0
    node.next[node.zeroLevel - 1] $\leftarrow$ node.next[node.zeroLevel]
    node.zeroLevel--
  return
    
fun updateUpToLevel(Node node, int level):
  node.lock.lock()
  while node.zeroLevel > level:
    updateZeroLevel(node)
  node.lock.lock()
  return
\end{lstlisting}
\vspace{-0.5cm}
\captionof{figure}{Lazy expansion functions}
\label{fig:auxiliary}

The method \texttt{update} that performs the balancing phase in forward pass is presented on Figure~\ref{fig:update}.
\begin{lstlisting}
fun getHits(Node node, int h):
  if node.zeroLevel > h:
    return node.selfhits
  return node.selfhits + node.hits[h]

fun update(K key):
  currM $\leftarrow$ fetch_and_add(list.m)

  list.head.lock()
  list.head.hits[MAX_LEVEL]++
  Node pred $\leftarrow$ list.head
  for h $\leftarrow$ MAX_LEVEL-1 .. zeroLevel:
    while pred.zeroLevel > h:
      updateZeroLevel(pred)
    predpred $\leftarrow$ pred
    curr $\leftarrow$ pred.next[h]
    updateUpToLevel(curr, h)
    if curr.key > key:
      pred.hits[h]++
      continue
      
    found_key $\leftarrow$ false
    while curr.key $\leq$ key:
      updateUpToLevel(curr, h)
      acquired $\leftarrow$ false
      if curr.next[h].key > key:
        curr.lock.lock()
        if curr.next[h].key $\leq$ key:
          curr.lock.unlock()
        else:
          acquired $\leftarrow$ true
          if curr.key = key:
            curr.selfhits++
            found_key $\leftarrow$ true
          else:
            curr.hits[h]++
      // Ascent condition
      if h + 1 < MAX_LEVEL and h < predpred.topLevel and
          predpred.hits[h + 1] - predpred.hits[h] > $\frac{currM}{2^{MAX\_LEVEL - 1 - h - 1}}$:
        if not acquired:
          curr.lock.lock()
        curh $\leftarrow$ curr.topLevel
        while curh + 1 < MAX_LEVEL and curh < predpred.topLevel and
            predpred.hits[curh + 1] - predpred.hits[curh] >
                                          $\frac{currM}{2^{MAX\_LEVEL - 1 - curh - 1}}$:
          curr.topLevel++
          curh++
          curr.hits[curh] $\leftarrow$ predpred.hits[curh] -
              predpred.hits[curh - 1] - curr.selfhits
          curr.next[curh] $\leftarrow$ predpred.next[curh]
          predpred.hits[curh] $\leftarrow$ predpred.hits[curh - 1]
          predpred.next[curh] $\leftarrow$ curr
        predpred $\leftarrow$ curr
        pred $\leftarrow$ curr
        curr $\leftarrow$ curr.next[h]
        continue
      // Descent condition
      elif curr.topLevel = h and curr.next[h].key $\leq$ key and
          getHits(curr, h) + getHits(pred, h) $\leq \frac{currM}{2^{MAX\_LEVEL - 1 - h}}$:
        currZeroLevel $\leftarrow$ list.zeroLevel
        if pred $\neq$ predpred:
          pred.lock.lock()
        curr.lock.lock()
        // Check the conditions that nothing has changed
        if curr.topLevel $\neq$ h or
            getHits(curr, h) + getHits(pred, h) > $\frac{currM}{2^{MAX\_LEVEL - 1 - h}}$  or
            curr.next[h].key > key or pred.next[h] $\ne$ curr:
          if pred $\ne$ predpred:
            pred.lock.unlock()
          curr.lock.unlock()
          curr $\leftarrow$ pred.next[h]
          continue
        else:
          if h = currZeroLevel:
            CAS(list.zeroLevel, currZeroLevel, currZeroLevel - 1)
          if curr.zeroLevel > h - 1:
            updateZeroLevel(curr)
          if pred.zeroLevel > h - 1:
            updateZeroLevel(pred)
          pred.hits[h] $\leftarrow$ pred.hits[h] + getHits(curr, h)
          curr.hits[h] $\leftarrow$ 0
          pred.next[h] $\leftarrow$ curr.next[h]
          curr.next[h] $\leftarrow$ null
          if pred $\neq$ predpred:
            pred.lock.unlock()
          curr.topLevel--
          curr.lock.unlock()
          curr $\leftarrow$ pred.next[h]
          continue
      pred $\leftarrow$ curr
    if predpred $\neq$ pred:
      predpred.lock.unlock()
    if found_key:
      pred.lock.unlock()
      return
  pred.lock.unlock()
\end{lstlisting}
\vspace{-0.5cm}
\captionof{figure}{Pseudocode of the update function.}
\label{fig:update}

\section{Additional Experimental Results}
\label{sec:additional-results}

\subsection{Uniform workload: $10^5-100-100$}
We consider a uniform workload $10^5-100-100$, i.e., the arguments of \texttt{contains} operations are chosen uniformly at random (Figure~\ref{fig:100-100}). As expected we lose performance lose relative to the skip-list due to the additional work our data structure performs. 
Note also that the CBTree outperforms \Listname{} in this setting. This is also to be expected, since the access cost, i.e., the number of links to traverse, is less for the CBTree.

\begin{figure*}
\begin{subfigure}{.33\textwidth}
\includegraphics[width=\linewidth]{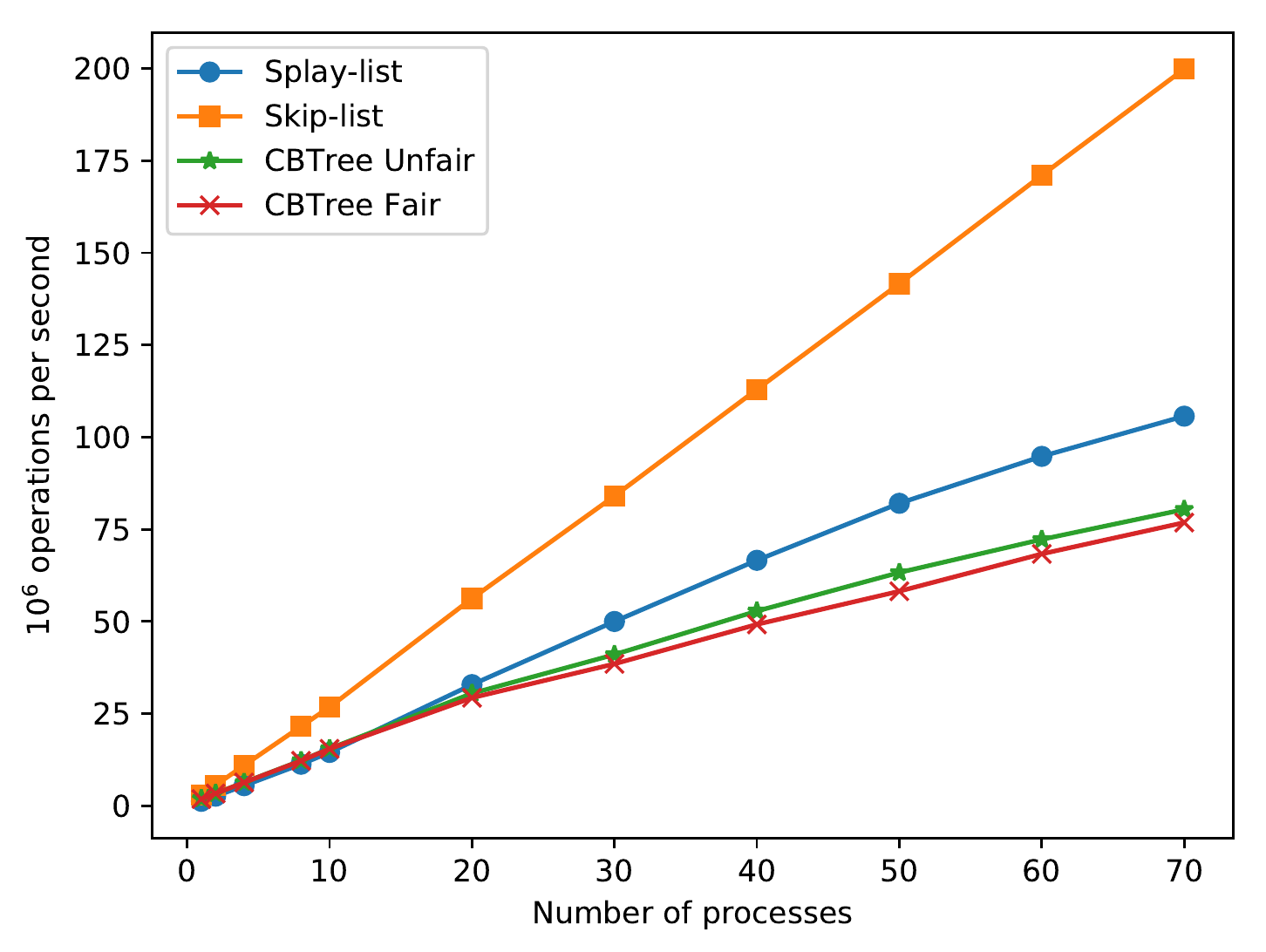}
\caption{$p = \nicefrac{1}{10}$}
\end{subfigure}
\begin{subfigure}{.33\textwidth}
\includegraphics[width=\linewidth]{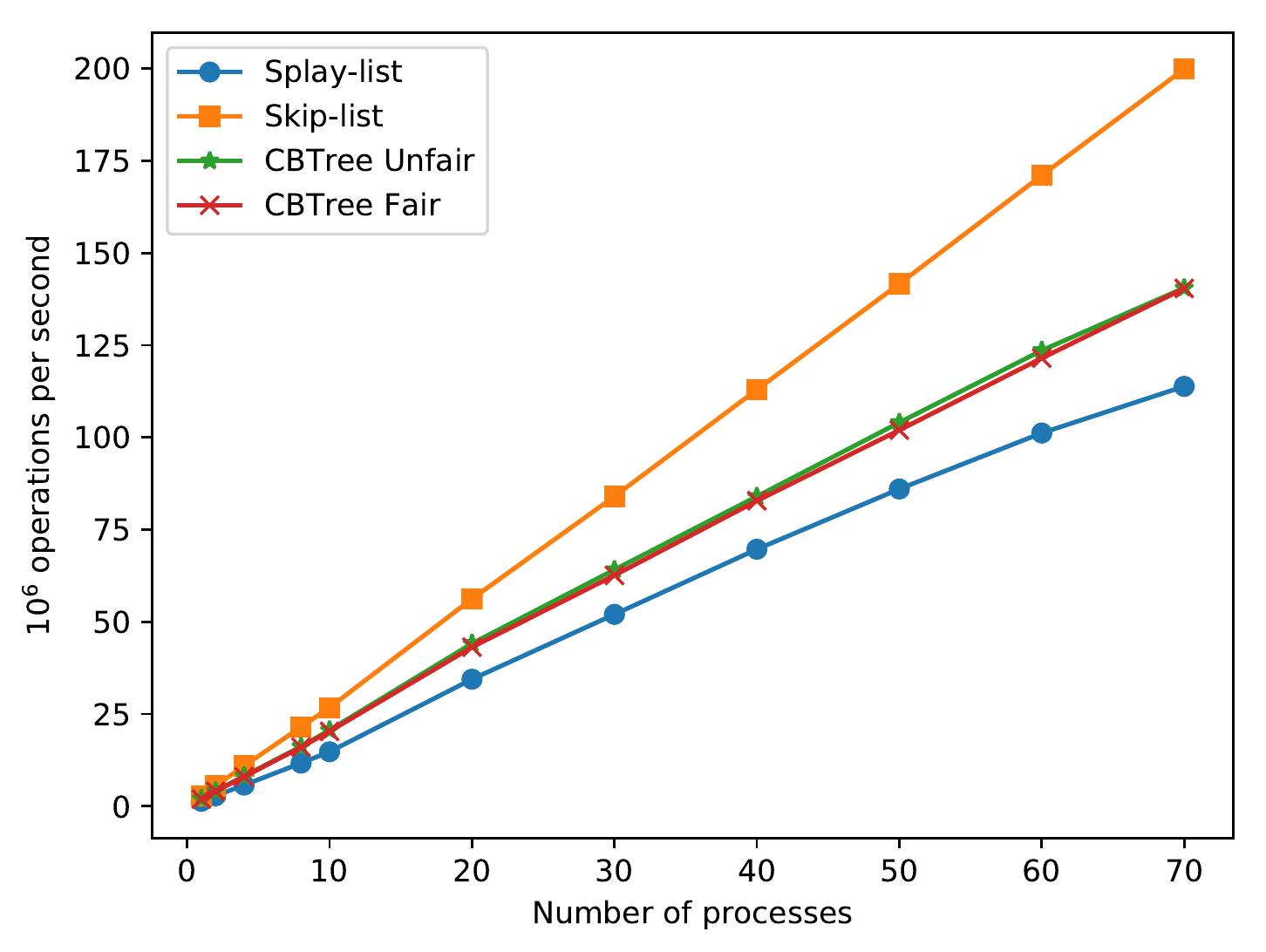}
\caption{$p = \nicefrac{1}{100}$}
\end{subfigure}%
\begin{subfigure}{.33\textwidth}
\includegraphics[width=\linewidth]{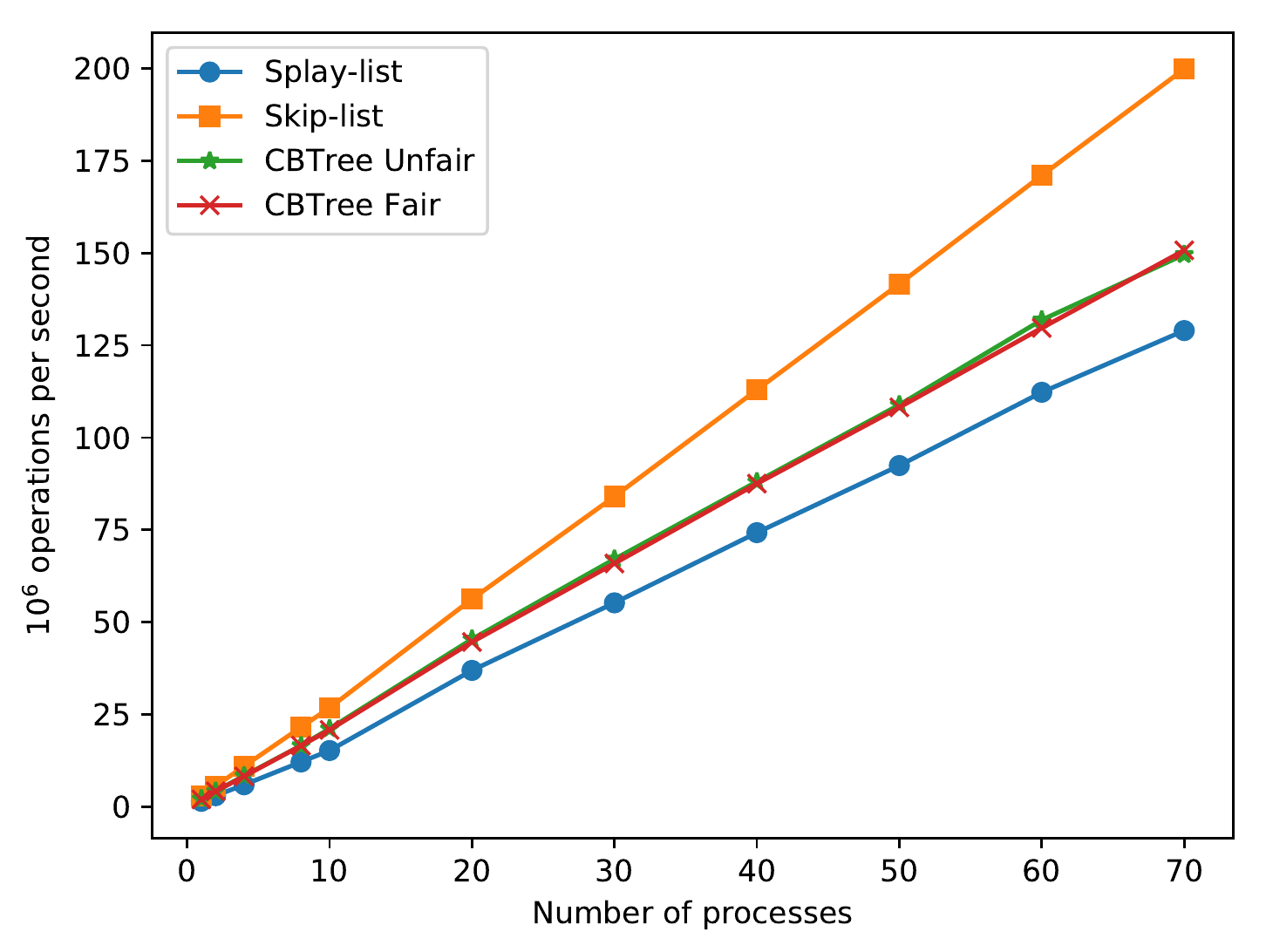}
\caption{$p = \nicefrac{1}{1000}$}
\end{subfigure}
\caption{Concurrent throughput for uniform workload.}
\label{fig:100-100}
\end{figure*}

\subsection{Zipf Distribution}
We also ran the data structures on an input coming from a Zipf distribution with the skew parameter set to $1$, which is the standard value: for instance, the frequency of words in the English language satisfies this parameter. As one can see on Figure~\ref{fig:zipf}, our \listname{} outperforms or matches all other data structures. 

\begin{figure*}
\begin{subfigure}{.33\textwidth}
\includegraphics[width=\linewidth]{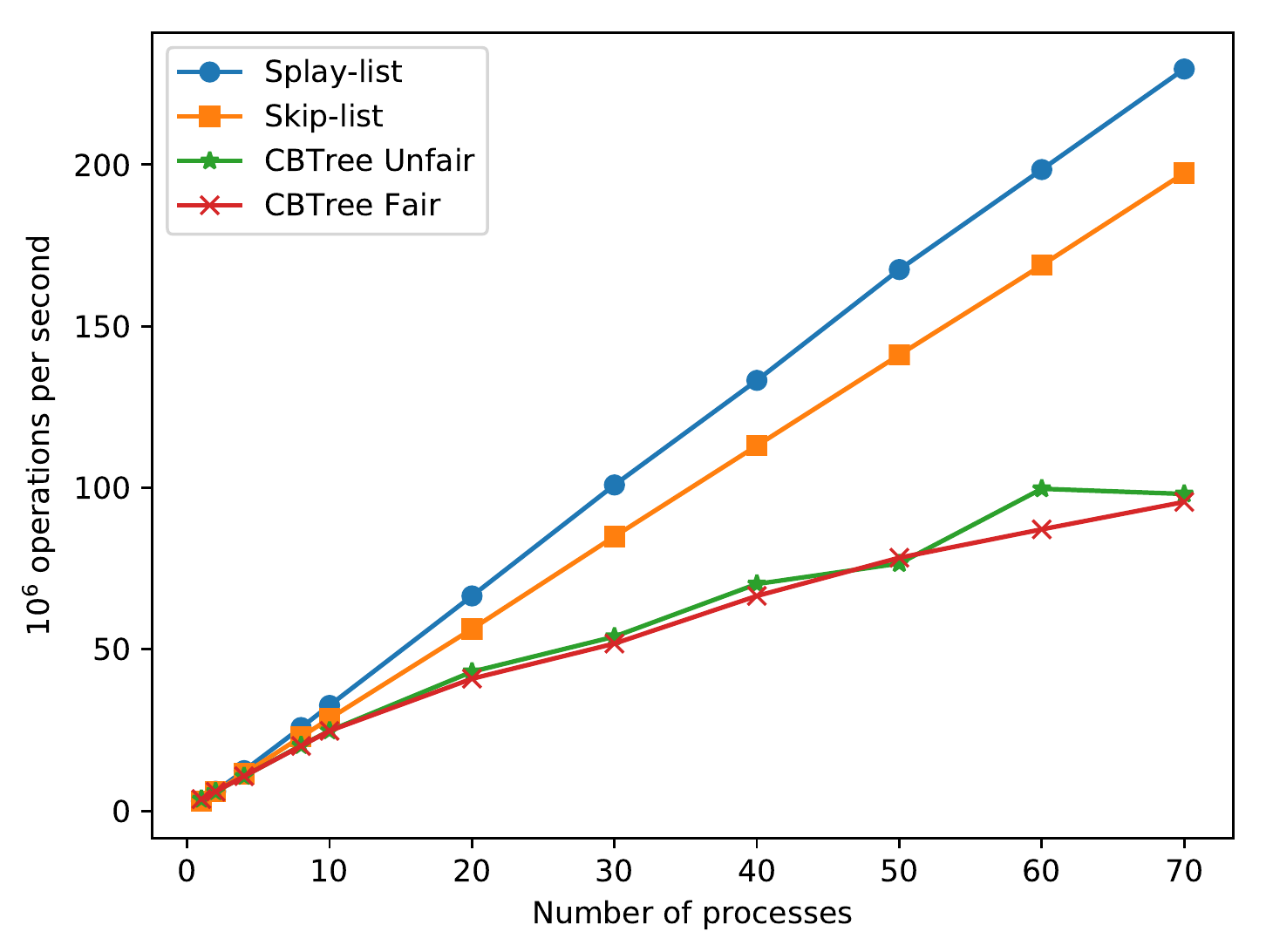}
\caption{$p = \nicefrac{1}{10}$}
\end{subfigure}
\begin{subfigure}{.33\textwidth}
\includegraphics[width=\linewidth]{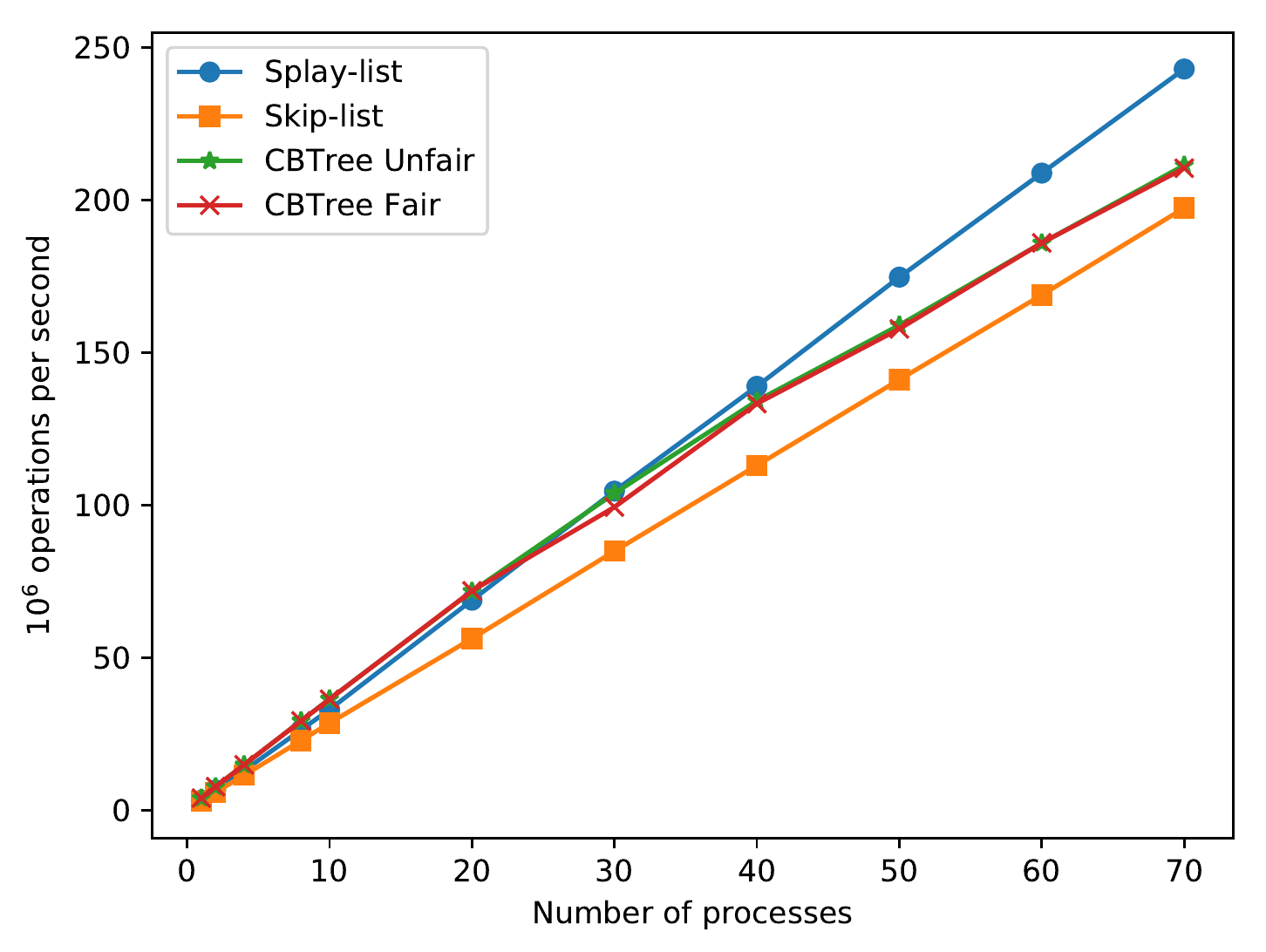}
\caption{$p = \nicefrac{1}{100}$}
\end{subfigure}%
\begin{subfigure}{.33\textwidth}
\includegraphics[width=\linewidth]{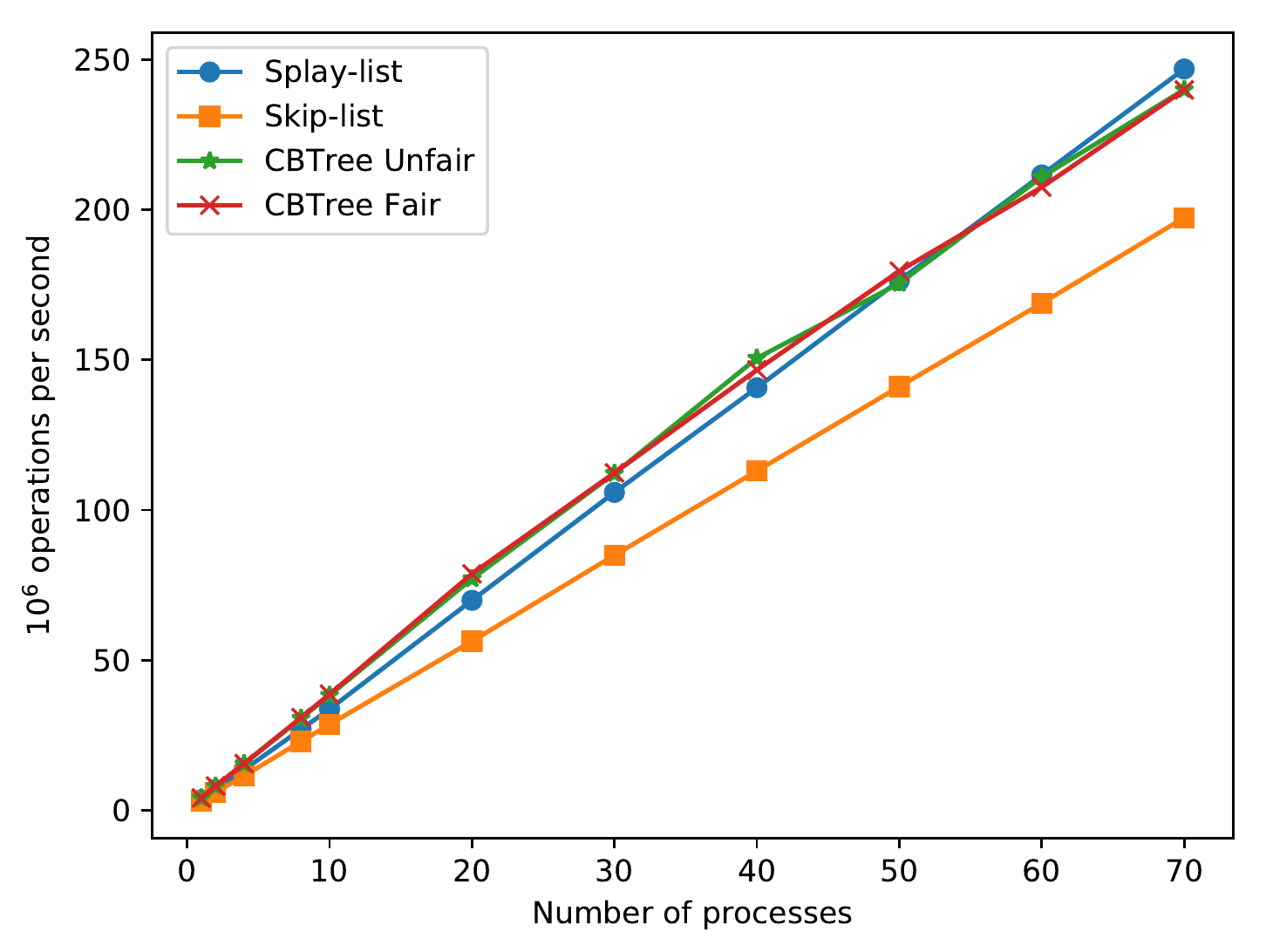}
\caption{$p = \nicefrac{1}{1000}$}
\end{subfigure}
\caption{Concurrent throughput on Zipf $1$ 
workload.}
\label{fig:zipf}
\end{figure*}

\subsection{General workloads}
In addition to read-only workloads we implemented general workloads, allowing for inserts and deletes, in our framework. General workloads are specified by five parameters $n-r-x-y-s$: 
\begin{enumerate}
\item ~$n$, the size of the workset of keys; \item ~$r\%$, the amount of \texttt{contains} performed; \item ~$x\%$ of \texttt{contains} are performed on $y\%$ of keys; \item ~\texttt{insert} and \texttt{delete} chooses a key uniformly at random from $s\%$ of keys.
\end{enumerate}

More precisely, we choose $n$ keys as set $S$ and we pre-populate the \listname{}: we add a key from $S$ with probability $00\%$. Then, we choose $s \cdot n$ keys uniformly at random to get $W$ key set. Also, we choose $y \cdot n$ keys from \emph{inserted} keys to get $R$ key set. We start $T$ threads, each of which chooses an operation: with probability $r\%$ it chooses \texttt{contains} and with probabilities $\frac{100-r}{2}\%$ it chooses \texttt{insert} or \texttt{delete}. Now, the thread has to choose an argument of the operation: for \texttt{contains} operation it chooses an argument from $R$ with probability $x\%$, otherwise, it chooses an argument from $S \setminus R$; for \texttt{insert} and \texttt{delete} operations it chooses an argument from $W$ uniformly at random.

We did not perform a full comparison with all other data structures (skip-list and the CBTree). However, we did a comparison to the splay-list iteself on the following two types of workloads: read-write workloads, $10^5-98-90-10-25$, $10^5-98-95-5-25$ and $10^5-98-99-1-25$~--- choosing \texttt{contains} operation with probability $98\%$, and \texttt{insert} and \texttt{delete} operations takes one quarter of elements as arguments; and read-only workloads, $10^5-0-90-10-0$, $10^5-0-95-5-0$ and $10^5-0-99-1-0$~--- read-only workload. 

The intuition is that the \listname{} should perform better on the second type of workloads, but by how much? We answer this question: the overhead does not exceed $15\%$ on $99-1$-workloads, does not exceed $7\%$ on $95-5$-workloads, and does not exceed $5\%$ on $90-1$-workloads.  As expected, the less a workload is skewed, the less the overhead. By that, we obtain that the small amount of \texttt{insert} and \texttt{delete} operations does not affect the performance significantly.

\subsection{Longer executions}
We run the \listname{} with the best parameter $p = \frac{1}{100}$ for ten minutes on one process on the following distributions: $10^5-90-10$, $10^5-95-5$, $10^5-99-1$ and Zipf with parameter $1$. Then, we compare the measured throughput per second with the throughput per second on runs of ten seconds. Obviously, we expect that the throughput increases since the data structure learns more and more about the distribution after each operation. And it indeed happens as we can see on Table~\ref{table:10mins}. In the long run, the improvement is up to 30\%.

\begin{table}
\center
\begin{tabular}{|c|c|c|}
\hline
Distribution & 10 sec & 10 min \\\hline
$10^5-90-10$ & 2777150 & 3630640 (+30\%) \\\hline
$10^5-95-5$  & 3401220 & 4403906 (+29\%) \\\hline
$10^5-99-1$  & 6707690 & 8184215 (+22\%) \\\hline
Zipf $1$     & 3806500 & 4261981 (+12\%) \\\hline
\end{tabular}
\caption{Comparison of the throughput on runs for 10 seconds and 10 minutes}
\label{table:10mins}
\end{table}

\subsection{Correlation between Key Popularity and Height}
We run the \listname{} with the best parameter $p = \frac{1}{100}$ for $100$ seconds on one process on the following distributions: $10^5-90-10$, $10^5-95-5$, $10^5-99-1$ and Zipf with parameter $1$. Then, we build the plots (see Figure~\ref{fig:correlation}) where for each key we draw a point ($x$, $y$) where $x$ is the number of operations per key and $y$ is the height of the key. We would expect that the larger the number of operations, the higher the nodes will be. This is obviously the case under Zipf distribution. With other distributions the correlation is not immediately obvious, however, one can see that if the number of operations per key is high, then the lowest height of the key is much higher than $1$.

\begin{figure*}
\begin{tabular}{c c}
\begin{subfigure}{.5\textwidth}
\includegraphics[width=\linewidth]{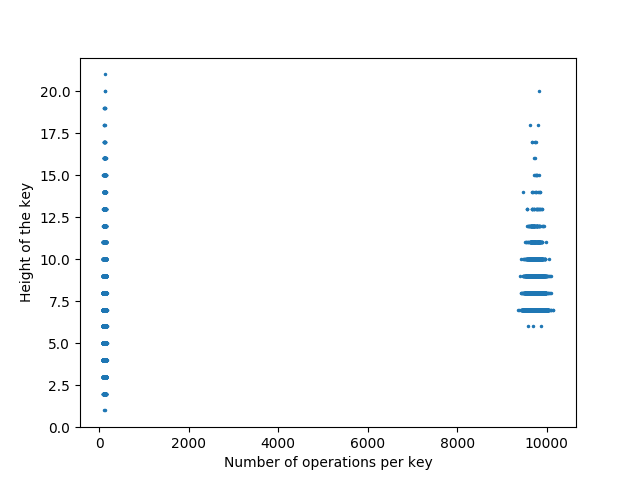}
\caption{Distribution $10^5-90-10$}
\end{subfigure} &
\begin{subfigure}{.5\textwidth}
\includegraphics[width=\linewidth]{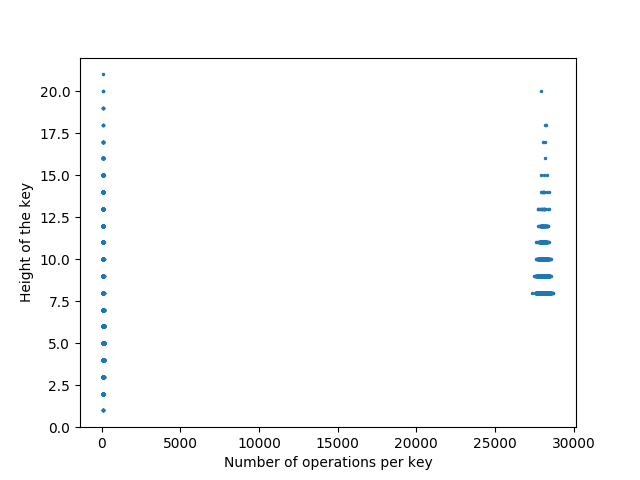}
\caption{Distribution $10^5-95-5$}
\end{subfigure} \\
\begin{subfigure}{.5\textwidth}
\includegraphics[width=\linewidth]{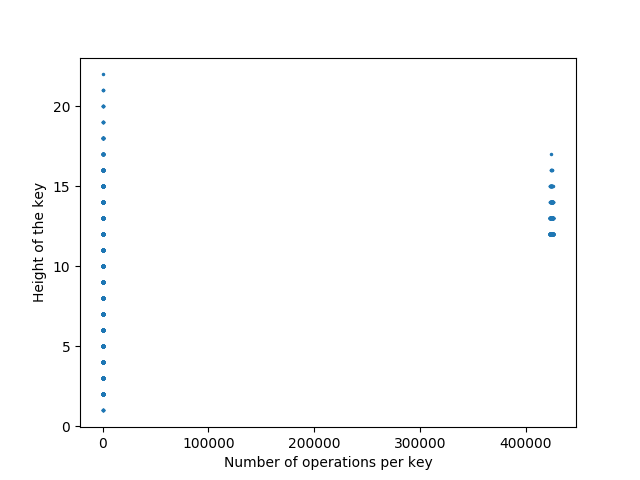}
\caption{Distribution $10^5-99-1$}
\end{subfigure} &
\begin{subfigure}{.5\textwidth}
\includegraphics[width=\linewidth]{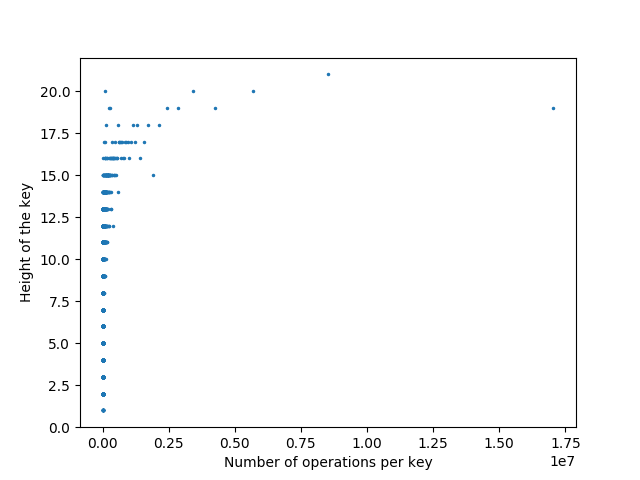}
\caption{Zipf distribution with parameter $1$}
\end{subfigure}
\end{tabular}
\caption{The correlation between the popularity and the height}
\label{fig:correlation}
\end{figure*}

\end{document}